\documentclass[11pt]{article} 

\usepackage[utf8]{inputenc} 

\usepackage{cite}

\usepackage{hyperref}
\hypersetup{colorlinks=false}

\usepackage{geometry,amsmath} 
\usepackage{graphicx} 

\usepackage{booktabs} 
\usepackage{array} 
\usepackage{paralist} 
\usepackage{verbatim} 
\usepackage{subfig} 

\usepackage{float}

\usepackage{todonotes}

\usepackage{amssymb}
\usepackage{amsmath}
\usepackage{amsthm}
\usepackage{mathrsfs}

\usepackage{xcolor,pict2e}

\theoremstyle{plain}
\newtheorem{thm}{Theorem}[section]
\newtheorem{cor}[thm]{Corollary}
\newtheorem{lem}[thm]{Lemma}
\newtheorem{prop}[thm]{Proposition}
\newtheorem{exm}[thm]{Example}
\newtheorem{defn}[thm]{Definition}
\newtheorem{rem}[thm]{Remark}

\usepackage{fancyhdr} 
\usepackage{sectsty}
\allsectionsfont{\sffamily\mdseries\upshape} 

\usepackage[nottoc,notlof,notlot]{tocbibind} 
\usepackage[titles,subfigure]{tocloft} 

\title{Symmetries, conservation laws and Noether's theorem for differential-difference equations}
\author{Linyu Peng\footnote{Email: L.Peng@aoni.waseda.jp}\vspace{0.4cm}
\\
{\it Department of Applied Mechanics and Aerospace Engineering, Waseda University, }\\ {\it Ohkubo, Shinjuku-ku, Tokyo 169-8555, Japan}}

\date{} 

\begin{document}
\maketitle

\abstract{
This paper mainly contributes to the extension of Noether's theorem to differential-difference equations. For that purpose, we first investigate the prolongation formula for continuous symmetries, which makes a characteristic representation possible. The relations of symmetries, conservation laws and the Fr\'echet derivative are also investigated. For non-variational equations, since Noether's theorem is now available, the self-adjointness method is adapted to the computation of conservation laws for differential-difference equations. A couple of differential-difference equations are investigated as illustrative examples, including the Toda lattice and semi-discretisations of the Korteweg-de Vries (KdV) equation. In particular, the Volterra equation is taken as a running example.
}
\vspace{0.2cm}


\section{Introduction}
The power of symmetry methods for differential equations has mainly been revealed during the last century; even though Sophus Lie's pioneering work was done in the 1880s (e.g. \cite{Ackerman1975,AH1976}).  In the popularisation of the applications of Lie's method to differential equations, Peter Olver's book \cite{Olver1993} is certainly one of the most inspiring and comprehensive masterpieces. In particular, great advances of Noether's (first and second) theorems \cite{Noether1918} are discussed therein, either for the broadest possibility of symmetries for variational problems and the corresponding conservation laws for Euler-Lagrange equations, or for infinite dimensional symmetries yielding differential relationships among Euler-Lagrange equations. Using different language, Vinogradov also made great developments in Noether theory at almost the same era (e.g. \cite{Kra1999,Vinogradov1984}). For a complete history, the reader should refer to Kosmann-Schwarzbach's book \cite{Kosmann2011}. For the interest of the broader community, other excellent textbooks are also available, see for example  \cite{Bluman1989,Hydon2000,Ibragimov1985,Mansfield2010,Marsden1999,Olver1995,Stephani1989}.

For (finite) difference equations, the methods of continuous symmetries and conservation laws, as well as their applications have been well adapted during the last few decades, see, for instance\cite{Dorod2010,Grant2013,Hydon2014,Levi2006,Marsden2001,Mikhailov2011,Quispel1984,Rasin2007a,Rasin2007b,Xenitidis2012}. The first effort to understand Noether's theorems was made by Maeda \cite{Maeda1981}, who considered symmetries of discrete mechanical systems. Since then, Noether's theorems for difference variational problems have been extensively investigated, see, for example \cite{Dorod2001,Hydon2004} for Noether's first theorem, and \cite{Hydon2011} for Noether's second theorem.
Surprisingly, a general differential-difference version of Noether's theorems has not yet been available, while the investigations of continuous symmetries and conservation laws (e.g. \cite{GH1999,Goktas1997,Jiang1998,Li2008,Shen2010}), integrability (e.g. \cite{AC1991,Cherdantsev1995,Yamilov2006}) and variational property (e.g. \cite{Kuper1985}) of differential-difference equations (DDEs) started relatively early. To fill this gap, the first purpose of this paper is to extend Noether's (first) theorem to DDEs. 

For non-variational differential or difference equations, the self-adjointness method has shown its efficiency in the computation of conservation laws from symmetries of given differential or difference equations, see, for instance \cite{Ibragimov2007} for a differential version and \cite{Peng2015} for a difference version. The adjoint system comes from the Euler-Lagrange equations governed by a formal Lagrangian, which is defined using the original system. A system of differential or difference equations is said to be self-adjoint if, via a substitution, the adjoint system holds on all solutions of the original system. In this paper we extend the self-adjointness method to DDEs. As Noether's theorem is available, it is immediately applicable to various DDEs, in particular to some of those from Yamilov's classification of integrable DDEs \cite{Yamilov2006}.

To commence our discussion for DDEs, it is helpful to begin with the counterparts for differential and difference equations, which are introduced in Section 2. We provide a simple proof of Noether's theorem for both cases by using the corresponding Fr\'echet derivatives. The self-adjointness method is summarised. It is already well understood that the continuous KdV equation is self-adjoint \cite{Ibragimov2007}; we show here that the discrete KdV equation is also self-adjoint. 
In Section 3, we start with the prolongation formula of continuous generalised symmetries, and conclude that, in the most general cases, only prolongations of regular infinitesimal generators can be equivalently expressed of the evolutionary form. Noether's theorem then follows for regular variational symmetries. The relation of the Fr\'echet derivative and symmetries or conservation laws is also investigated. Section 4 is devoted to the self-adjointness method and its application to the computation of conservation laws of non-variational DDEs. Illustrative examples are provided, especially integrable DDEs.

\section{Noether's theorem for differential \& difference Equations}
In this section, we review relevant facts related with Noether's theorem. We also prove several new results for difference equations.

\subsection{Differential equations}

We first consider a differential version of Noether's theorem. We mainly follow Olver's language \cite{Olver1993}. Other practical textbooks are also available, for instance \cite{Bluman1989,Hydon2000,Ibragimov1985}. In particular, Kosmann-Schwarzbach's book \cite{Kosmann2011} includes a complete history of Noether's (two) theorems.  

 Let $x=(x^1,x^2\ldots,x^p)\in \mathbb{R}^p$ be the independent variable and let $u=(u^1,u^2,\ldots,u^q)\in\mathbb{R}^q$ be the dependent variable. Partial derivatives of $u^{\alpha}$ are written in the multi-index form $u_J^{\alpha}$ where $J=(j_1,j_2,\ldots,j_p)$. For the differential case, each index $j_i$ is a non-negative integer which denotes the number of derivatives with respect to $x^i$. Namely
\begin{equation}
u^{\alpha}_J=\frac{\partial^{|J|} u^{\alpha}}{\partial (x^1)^{j_1}\partial (x^2)^{j_2}\ldots \partial (x^p)^{j_p}},
\end{equation}
where $|J|=j_1+j_2+\cdots +j_p$. Locally, a system of differential equations can be written as
\begin{equation}\label{sode}
\mathcal{A}=\{F_{\alpha}(x,[u])=0\},
\end{equation}
where $[u]$ is shorthand for $u$ and finitely many of its partial derivatives. Consider a one-parameter group of Lie point transformations as follows:
\begin{equation}\label{lptde}
\widetilde{x}=\widetilde{x}(\varepsilon;x,u),\quad \widetilde{u}=\widetilde{u}(\varepsilon;x,u),
\end{equation}
subject to $\widetilde{x}(e;x,u)=x$, $\widetilde{u}(e;x,u)=u$. Here $\varepsilon=e$ is the identity element of the one-parameter group.
Define the {\bf total derivative} with respect to $x^i$ as 
\begin{equation}
D_i=\frac{\partial}{\partial x^i}+\sum_{\alpha,J}u_{J+\bold{1}_i}^{\alpha}\frac{\partial}{\partial u_J^{\alpha}},
\end{equation}
where $\bold{1}_i$ is the
$p$-tuple with only one nonzero entry $1$ in the $i$-th place. In this paper, we often write the first-order derivatives $u_{\bold{1}_i}$ (and first-order forward shifts for difference equations) as $u_i$ for simplicity.
The prolonged transformations to higher-order can then be immediately obtained via the chain rule, e.g.
\begin{equation}
\widetilde{u_{\bold{1}_i}^{\alpha}}=\frac{\partial\widetilde{u}^{\alpha}}{\partial \widetilde{x}^i}=\sum_j\frac{D_j\widetilde{u}^{\alpha}}{D_j\widetilde{x}^i}.
\end{equation}
The corresponding {\bf infinitesimal generator} is 
\begin{equation}
X=\xi^i(x,u)\partial_{x^i}+\phi^{\alpha}(x,u)\partial_{u^{\alpha}},
\end{equation}
where 
\begin{equation}
\xi^i=\frac{\operatorname{d}\! \widetilde{x}^i}{\operatorname{d}\!\varepsilon}\Big|_{\varepsilon=e},\quad \phi^{\alpha}=\frac{\operatorname{d}\!\widetilde{u}^{\alpha}}{\operatorname{d}\! \varepsilon}\Big|_{\varepsilon=e}.
\end{equation}
[Note. The Einstein summation convention is used here and all throughout the paper.] The prolonged generator $\bold{pr}^{(k)}\!X$ can be written in terms of $u$, $\xi$, $\phi$ and their derivatives, by virtue of the prolonged transformations:
\begin{equation}
\bold{pr}^{(k)}\!X=X+\sum_{\alpha}\sum_{|J|=1}^k\phi_J^{\alpha}\frac{\partial}{\partial u_J^{\alpha}},
\end{equation}
where
\begin{equation}
\phi^{\alpha}_{J+\bold{1}_i}=D_i\phi_J^{\alpha}-(D_i\xi^j)u_{J+\bold{1}_j}^{\alpha}.
\end{equation}
\begin{defn}[Infinitesimal invariance criterion \cite{Olver1993}]
Suppose the differential system (\ref{sode}) satisfies the maximal rank condition. The transformations (\ref{lptde}) form a symmetry group of $\mathcal{A}$ if its infinitesimal generator $X$ satisfies the following {\bf linearized symmetry condition}, that is
\begin{equation}\label{lscde}
\bold{pr}^{(k)}X(F_{\alpha}(x,[u]))=0,\text{ whenever } \{F_{\alpha}(x,[u])=0\},
\end{equation}
where $k$ is a properly chosen finite integer. For the sake of simplicity, we will omit the index $k$ and just write $\bold{pr}X$ in the current paper.
\end{defn}
Such symmetries are usually called {\bf Lie point symmetries}. More generally, the infinitesimal generator may also depend on derivatives of $u$, namely $\xi=\xi(x,[u])$ and/or $\phi=\phi(x,[u])$, which we call {\bf generalised symmetries}. In both cases, it is often more convenient to write prolonged generators in terms of {\bf characteristics} of symmetries $Q^{\alpha}=\phi^{\alpha}-\xi^iu^{\alpha}_{\bold{1}_i}$, that is
\begin{equation}
\bold{pr}X=\xi^iD_i+\sum_{\alpha,J}(D_JQ^{\alpha})\frac{\partial}{\partial u_J^{\alpha}}.
\end{equation}
Here we use the shorthand that $D_J=D_1^{j_1}D_2^{j_2}\ldots D_p^{j_p}$. As can be seen from the linearized symmetry condition, equivalently it is sufficient to consider evolutionary infinitesimal generators only, namely
\begin{equation}
\bold{pr}X=\sum_{\alpha,J}(D_JQ^{\alpha})\frac{\partial}{\partial u_J^{\alpha}}.
\end{equation}
The $\xi^iD_i$ part is trivial.

\begin{rem}
In the rest of this section, when we mention a group of symmetries, we always mean a group of generalised symmetries whose generators are of the evolutionary form unless otherwise specified. 
\end{rem}

Various methods for calculating symmetries for differential equations are available. For concrete examples the reader should refer to, for instance\cite{Ackerman1975,Olver1993,Kra1999,Bluman1989,Hydon2000,Ibragimov1985,Stephani1989,Ovsiannikov1982}.
\begin{rem}\label{glo}
 It is worth mentioning that the linearized symmetry condition (for vector fields of the evolutionary form) is equivalent to the following condition that
 \begin{equation}
 (\bold{D}_F)_{\alpha\beta}Q^{\beta}=0, \text{ whenever } \{F_{\alpha}(x,[u])=0\},
 \end{equation}
 where $Q$ is the characteristic and the {\bf Fr\'echet derivative} with respect to a tuple of functions $F_{\alpha}(x,[u])$ is defined as
 \begin{equation}
 \bold{D}_F(P):=\frac{\operatorname{d}}{\operatorname{d}\!\varepsilon}\Big|_{\varepsilon=0}F(x,[u+\varepsilon P(x,[u])]).
 \end{equation}
 It has also been called the {\bf linearisation operator} (e.g. \cite{Kra1999}). In components, it is  \begin{equation}
 (\bold{D}_F)_{\alpha\beta}=\sum_J\frac{\partial F_{\alpha}}{\partial u_J^{\beta}}D_J.
 \end{equation}
\end{rem}
%

Another important ingredient for Noether's theorem is a {\bf conservation law}, which is a divergence expression about a $p$-tuple $P$,
\begin{equation}
\operatorname{Div}P=0
\end{equation}
which vanishes on all solutions of a differential system $\mathcal{A}$. The divergence operator is defined as
\begin{equation}
\operatorname{Div}P=D_iP^i.
\end{equation}
Those we often encounter in fluid dynamics and mechanics include, for instance, the conservations of mass, energy or momenta. If the system $\mathcal{A}$ in (\ref{sode}) is totally nondegenerate (see \cite{Olver1993}) or analytic (see \cite{Ibragimov1985}), a divergence $\operatorname{Div}P$ vanishes on all of its solutions if and only if there exist functions $B^{\alpha}_J(x,[u])$ such that 
\begin{equation}
\operatorname{Div}P=\sum_{\alpha,J}B^{\alpha}_J(D_JF_{\alpha}).
\end{equation}
In particular, a system of the Kovalevskaya form satisfies the nondegeneracy condition. This can then be integrated by parts, which yields
\begin{equation}
\operatorname{Div}P=\operatorname{Div}R+Q^{\alpha}F_{\alpha}.
\end{equation}
Here 
\begin{equation}
Q^{\alpha}=\sum_J(-D)_JB^{\alpha}_J, 
\end{equation}
where the adjoint operator of $D_J$ is 
\begin{equation}
(-D)_J=(-1)^{|J|}D_J.
\end{equation}
Replacing $P$ by $P-R$, we get an equivalent conservation law (with a difference of trivial conservation law $R$)
\begin{equation}
\operatorname{Div}P=Q^{\alpha}F_{\alpha}.
\end{equation}
This is called the {\bf characteristic form} of a conservation law, and $Q$ is called the {\bf characteristic} of the given conservation law. 

Consider an action
\begin{equation}
\mathscr{L}[u]=\int L(x,[u])\operatorname{d}\!x^1\wedge\operatorname{d}\!x^2\wedge\cdots\wedge\operatorname{d}\!x^p.
\end{equation}
The corresponding Euler-Lagrange equations are obtained by varying each $u^{\alpha}$ and then integrating by parts. These equations can be written more compactly using the {\bf Euler operator} $\bold{E}_{\alpha}$, and the Euler-Lagrange equations become
\begin{equation}
\bold{E}_{\alpha}(L)=0,
\end{equation}
where
\begin{equation}
\bold{E}_{\alpha}:=\sum_J(-D)_J\frac{\partial}{\partial u_J^{\alpha}}.
\end{equation}
\begin{rem}
A function $L(x,[u])$ is a null Lagrangian in the sense that $\bold{E}(L)\equiv0$ if and only if $L$ is a total divergence.
\end{rem}

\begin{rem}\label{ccld}
It has been proven geometrically that for normal equations, characteristic of a conservation law lies in the kernel of the adjoint of the corresponding Fr\'echet derivative (e.g. \cite{Kra1999}). In \cite{Olver1993}, Olver provided a simple proof using the fact that 
\begin{equation}
\bold{E}(Q\cdot F)=\bold{D}_Q^{\ast}(F)+\bold{D}_F^{\ast}(Q),
\end{equation}
where for tuples $A$, $B$ and $F$ of proper orders, the Fr\'echet derivative $\bold{D}$ and its adjoint operator satisfy the following identity 
\begin{equation}
A\bold{D}_F(B)=B\bold{D}_F^{\ast}(A)+\operatorname{Div}P
\end{equation}
for a certain tuple $P$. Assume $Q$ is the characteristic of some conservation law for a system $\{F_{\alpha}(x,[u])=0\}$, then $Q\cdot F$ is a null Lagrangian and hence $\bold{E}(Q\cdot F)\equiv0$. On solutions of the system, $\bold{D}_Q^{\ast}(F)$ also vanishes since the adjoint operator is given by
\begin{equation}
(\bold{D}^{\ast}_Q)_{\alpha\beta}=\sum_J(-D)_J\cdot \frac{\partial Q^{\beta}}{\partial u_J^{\alpha}}.
\end{equation}
Therefore, $\bold{D}_F^{\ast}(Q)=0$ holds on solutions of the system if $Q$ is a characteristic of one of its conservation laws.
This is, sometimes, called the multiplier method for classifying conservation laws. 
\end{rem}

For a system of Euler-Lagrange equations $\bold{E}(L)=0$, the associated Fr\'echet derivative is always self-adjoint (e.g. \cite{Olver1993}), namely
\begin{equation}
(\bold{D}_{\bold{E}(L)})^{\ast}=\bold{D}_{\bold{E}(L)}.
\end{equation}
Hence Remark \ref{ccld} in some sense implies Noether's theorem through the relation between characteristics of symmetries and conservation laws. It is also interesting to realise that self-adjointness of a Fr\'echet derivative is sufficient but not necessary for constructing a relation between symmetries and conservation laws. For instance, skew self-adjointness is also sufficient, namely $\bold{D}_F^{\ast}=-\bold{D}_F$ for a system $\{F_{\alpha}(x,[u])=0\}$. This can further be generalised to 
\begin{equation}
\bold{D}^{\ast}_F=B(x,[u])\bold{D}_F
\end{equation}
for some non-zero functions $B(x,[u])$. A simple example is the following equation
\begin{equation}
u_t-u_{xxx}=0.
\end{equation}
The Fr\'echet derivative is
\begin{equation}
\bold{D}=D_t-D_x^3,
\end{equation}
which satisfies $\bold{D}^{\ast}=-\bold{D}$.

A generalised vector field 
\begin{equation}
X=\xi^i(x,[u])\frac{\partial}{\partial x^i}+\phi^{\alpha}(x,[u])\frac{\partial}{\partial u^{\alpha}}
\end{equation}
is a (divergence) {\bf variational symmetry} of an action $\mathscr{L}[u]$ if and only if there exists a $p$-tuple $P$ such that
\begin{equation}
\bold{pr}X(L)+L\operatorname{Div}\xi=\operatorname{Div}P.
\end{equation}
Equivalently, this can be re-organised into evolutionary form 
\begin{equation}
\begin{aligned}
\sum_{\alpha,J}(D_JQ^{\alpha})\frac{\partial L}{\partial u_J^{\alpha}}&=\operatorname{Div}P-L\operatorname{Div}\xi-\xi^iD_iL\\
&=D_i(P^i-\xi^iL).
\end{aligned}
\end{equation}
It is now well known that variational symmetries of a functional $\mathscr{L}[u]$ are still symmetries for the associated Euler-Lagrange equations. {\bf Noether's theorem} for differential systems (e.g. \cite{Olver1993,Noether1918,Kra1999}) states that  the characteristic $Q$ of a variational symmetry is a characteristic of a conservation law for the Euler-Lagrange equations. Namely, there exists a $p$-tuple $P$ such that
\begin{equation}
\operatorname{Div}P=Q^{\alpha}\bold{E}_{\alpha}(L).
\end{equation}

\subsection{Difference equations}

Now we consider finite difference equations. For varying discrete steps, the reader should refer to, e.g. \cite{Dorod2010,Dorod2001}. In this paper, however, we are mainly focused on fixed discrete steps, in particular difference equations on $\mathbb{Z}^p$. In this situation, $n\in\mathbb{Z}^p$ is understood as a vector of independent variables while $u_n\in\mathbb{R}^q$ is a vector of dependent variables. A first construction of Noether's theorem for ordinary difference equations was studied by Maeda \cite{Maeda1981}; a general  extension has now been well studied, see for instance\cite{Hydon2014,Dorod2001,Hydon2004,Hickman2003,Mansfield2006,Peng2013}. In particular, Kupershmidt's book \cite{Kuper1985} includes fundamental analysis for variational principles with respect to difference as well as differential-difference equations.

More precisely, let $n=(n^1,n^2,\ldots,n^p)\in \mathbb{Z}^p$ and
$u_n=(u_n^1,u_n^2,\ldots,u_n^q)\in U\subset \mathbb{R}^q$ be the
independent and dependent variables respectively. The shift operator
(or map) $S$ is defined as
\begin{equation}
\begin{aligned}
S_k:n\mapsto n+\bold{1}_k,
\end{aligned}
\end{equation}
while its generalisation to a function $f(n)$ is
\begin{equation}
\begin{aligned}
S_k:f(n)\mapsto f(S_kn).
\end{aligned}
\end{equation}
We will also use the notation $S_{\bold{1}_k}=S_k$, and the
composite of shift operators using multi-index notation is given by $S_J=S_1^{j_1}S_2^{j_2}\ldots S_p^{j_p}$,
where $J=(j_1,j_2,\ldots,j_p)$ is a $p$-tuple. However, different from the differential case, now each index is $j_i\in\mathbb{Z}$. The inverse of the shift operator $S_{\bold{1}_k}$ is given by
 $S_{-\bold{1}_k}:n\to n-\bold{1}_k$. The inverse operator $S_{-J}$ of
the composite of shifts is similarly denoted.

Consider a one-parameter (Lie point) transformations, which keeps $n$ invariant, namely
\begin{equation}
\widetilde{u}_n=\widetilde{u}(\varepsilon;n,u_n) \text{ and } \widetilde{u}_n|_{\varepsilon=e}=u_n.
\end{equation}
The corresponding {\bf infinitesimal generator} is 
\begin{equation}
X=Q^{\alpha}_n(n,u_n)\frac{\partial}{\partial u_n^{\alpha}},
\end{equation}
where  the {\bf characteristic} $Q_n=(Q_n^1,Q_n^2,\ldots,Q^q_n)$ is
\begin{equation}
Q_n^{\alpha}(n,u_n)=\frac{\operatorname{d}\!\widetilde{u}_n^{\alpha}}{\operatorname{d}\! \varepsilon}\Big|_{\varepsilon=e}.
\end{equation}
For generalised (or higher-order)
transformations , the characteristic is $Q_n^{\alpha}=Q_n^{\alpha}(n,[u])$ where $[u]$ denotes $u_n$ and a finite number of 
its shifts in the difference case. Prolongations
of the vector field $X$ can be immediately achieved via the identities 
\begin{equation}
S_J\widetilde{u}_n=\widetilde{u}_{n+J}
\end{equation}
and we have
\begin{equation}
\bold{pr}X=\sum_{\alpha,J}\left(S_JQ^{\alpha}_n\right)\frac{\partial}{\partial
u_{n+J}^{\alpha}}.
\end{equation}
For the sake of
simplicity, we also write 
$u^{\alpha}_{n+J}$ as $u^{\alpha}_J$. When $J$ vanishes, we sometimes omit the indices, e.g. $Q^{\alpha}=Q_n^{\alpha}$ and $u^{\alpha}=u^{\alpha}_n$. 


\begin{defn}
Consider a system of difference equations
\begin{equation}
\mathcal{A}^{\vartriangle}=\{F_{\alpha}(n,[u])=0\}. \label{fde}
\end{equation}
 A one-parameter group of generalised transformations  is a symmetry group if the associated
infinitesimal generator satisfies the following {\bf linearized symmetry condition}
\begin{equation}
\bold{pr}X(F_{\alpha})=0,
\end{equation}
when (\ref{fde}) holds.
\end{defn}

 For a system of
difference equations given by (\ref{fde}), a {\bf conservation law}
is defined as a {\bf difference divergence} expression
\begin{equation}\label{defcl}
\operatorname{Div}^{\vartriangle}P(n,[u]):=\sum_{i=1}^p(S_i-\operatorname{id})P^i(n,[u]),
\end{equation}
which vanishes on all solutions of the system. Here $P$ is a $p$-tuple and the operator $\operatorname{id}$ is the identity operator. For a difference system of the Kovalevskaya form (see, e.g. \cite{Grant2013,Hydon2014,Peng2013}) , there exist functions $B^{\alpha}_J(n,[u])$ such that 
\begin{equation}\label{deltacl}
\operatorname{Div}^{\vartriangle}P(n,[u])=\sum_{\alpha,J}B^{\alpha}_J(S_JF_{\alpha}).
\end{equation}
In the difference case, integration by parts is replaced by the formula of summation by parts 
\begin{equation}\label{sbps}
(Sf)g=(S-\operatorname{id})(fS_{-1}g)+fS_{-1}g=\operatorname{Div}^{\vartriangle}(fS_{-1}g)+fS_{-1}g
\end{equation}
for any functions $f(n,[u])$ and $g(n,[u])$. Here the operator $S_{-1}$ is the inverse (or adjoint) of $S$; it is also called the backward shift operator. Hence, the difference conservation law becomes 
\begin{equation}
\operatorname{Div}^{\vartriangle}P=\operatorname{Div}^{\vartriangle}R+Q^{\alpha}F_{\alpha}.
\end{equation}
Similarly as in the differential case, the $p$-tuple $R$  also contributes to a trivial conservation law and if we replace $P$ by $P-R$, the conservation law is written in the {\bf characteristic form}
\begin{equation}
\operatorname{Div}^{\vartriangle}P=Q^{\alpha}F_{\alpha},
\end{equation}
where the {\bf characteristic} $Q^{\alpha}$ is given by
\begin{equation}
Q^{\alpha}(n,[u])=\sum_{J}S_{-J}B_J^{\alpha}(n,[u]).
\end{equation}
The operators $S_{-J}$ and $S_J$ are adjoint to each other with respect to the $\ell_2$ inner product. 
 [Note. In \cite{Mikhailov2011}, characteristics of conservation laws were called cosymmetries.]



Consider a difference variational problem (or an action)
\begin{equation}\label{dl1}
\mathscr{L}[u]=\sum_nL_n=\sum_nL(n,[u]).
\end{equation}
Its variation with respect to $u$ amounts to the discrete Euler-Lagrange equations 
\begin{equation}
\bold{E}_{\alpha}^{\vartriangle}(L_n)=0,
\end{equation}
where the {\bf discrete Euler operator} $\bold{E}^{\vartriangle}$ is defined by
\begin{equation}
\bold{E}_{\alpha}^{\vartriangle}=\sum_JS_{-J}\frac{\partial}{\partial u_J^{\alpha}}.
\end{equation}
\begin{rem}
A function $L(n,[u])$ contributes to a null Lagrangian, i.e. $\bold{E}_{\alpha}^{\vartriangle}\equiv0$, if and only if $L_n$ is of total divergence form, namely $L_n=\operatorname{Div}^{\vartriangle}P$. It is also worth mentioning that variational symmetries are still symmetries of the underlying Euler-Lagrange equations (see, e.g. \cite{Peng2014}); this is not true when discrete steps are not fixed (see, e.g. \cite{Dorod2010,Dorod2001}).
\end{rem}
The action $\mathscr{L}[u]$ is invariant with respect to a local group of (continuous) transformations if 
\begin{equation}
\sum_nL(n,[\widetilde{u}])=\sum_nL(n,[u]).
\end{equation}
 The group is then called a group of (divergence) {\bf variational symmetries} for the discrete action $\mathscr{L}[u]$. Infinitesimally, it is equivalent to the existence of a $p$-tuple $R$ such that
\begin{equation}\label{scd}
\bold{pr}X(L_n)=\operatorname{Div}^{\vartriangle}R,
\end{equation}
where $X$ is the infinitesimal generator of a symmetry group.
 A {\bf discrete version of
Noether's theorem} exists (see, e.g.
\cite{Hydon2014,Maeda1981,Dorod2001,Mansfield2006,Peng2013}); it also relates symmetries for
difference variational problems and conservation laws of the
associated Euler-Lagrange equations via their characteristics. Namely, there exists a $p$-tuple $P$ such that the characteristic $Q^{\alpha}$ of a variational symmetry satisfies
\begin{equation}
\operatorname{Div}^{\vartriangle}P=Q^{\alpha}\bold{E}_{\alpha}^{\vartriangle}(L_n).
\end{equation}

For a difference system (\ref{fde}), the linearized symmetry condition is equivalent to the following condition (e.g. \cite{Peng2014})
\begin{equation}
0=\bold{D}_F^{\vartriangle}(Q):=\frac{\operatorname{d}}{\operatorname{d}\!\varepsilon}\Big|_{\varepsilon=0}F(n,[u+\varepsilon Q(n,[u])]),
\end{equation}
where $Q$ is the characteristic and $\bold{D}^{\vartriangle}$ is the {\bf Fr\'echet derivative} in the difference case. Its adjoint is then
\begin{equation}
(\bold{D}_F^{\vartriangle})^{\ast}_{\alpha\beta}=\sum_{J}S_{-J}\cdot \frac{\partial F_{\beta}}{\partial u_J^{\alpha}}.
\end{equation}
\begin{lem}[\cite{Peng2014}]\label{lempeng}
For two $r$-tuples $A(n,[u])$ and $B(n,[u])$, the following identity holds
\begin{equation}
\bold{E}^{\vartriangle}(A\cdot B)=(\bold{D}^{\vartriangle}_A)^{\ast}(B)+(\bold{D}^{\vartriangle}_B)^{\ast}(A).
\end{equation}
\end{lem}
This lemma leads to a difference counterpart to the conclusion in Remark \ref{ccld}.
\begin{thm}
For a system of difference equations $\{F_{\alpha}(n,[u])=0\}$ of the Kovalevskaya form, characteristics $Q$ of its conservation laws lie in the kernel of the adjoint of the corresponding Fr\'echet derivative $\bold{D}_F^{\vartriangle}$ restricted to solution sections, namely 
\begin{equation}
(\bold{D}^{\vartriangle}_F)^{\ast}(Q)=0 \text{ whenever } \{F_{\alpha}(n,[u])=0\}.
\end{equation}
\end{thm}
\begin{proof}
Choose $A=Q$ and $B=F$ in Lemma \ref{lempeng}. The theorem is then verified since $Q\cdot F$ is a null Lagrangian and $(\bold{D}^{\vartriangle}_Q)^{\ast}(F)=0$ on solutions of the difference system.
\end{proof}
\begin{rem}
The Fr\'echet derivative of an Euler-Lagrange system $\bold{E}^{\vartriangle}(L_n)=0$ is always self-adjoint (e.g. \cite{Hydon2004,Peng2014,Kuper1985}), namely 
\begin{equation}
(\bold{D}^{\vartriangle}_{\bold{E}^{\vartriangle}(L_n)})^{\ast}=\bold{D}^{\vartriangle}_{\bold{E}^{\vartriangle}(L_n)}.
\end{equation}
 This hence implies the {\bf discrete version of Noether's theorem} through identifying characteristics of variational symmetries (hence symmetries of the Euler-Lagrange equations) and conservation laws. 
\end{rem}

\begin{exm}
The difference equation
\begin{equation}
u_2=\frac{u_1^2}{u}
\end{equation}
admits the following characteristics of symmetries (e.g. \cite{Hydon2014})
\begin{equation}
Q_1=u_0,\quad Q_2=nu_0,\quad Q_3=u_n\ln|u|.
\end{equation}
In \cite{Peng2015}, the equation was solved by obtaining all of its invariant first integrals with respect to the first two generators. It was also noted therein that through a change of variables $v_0=\ln u_0$, the resulting equation is governed by a Lagrangian
\begin{equation}
L_n=vv_1-v^2.
\end{equation}
The Euler-Lagrange equation is $v_1-2v+v_{-1}=0$.
The first two generators become variational symmetries of $L_n$:
\begin{equation}
X_1=\partial_{v},\quad X_2=n\partial_{v}.
\end{equation}
Therefore, we obtain two conservation laws in the characteristic form
\begin{equation}
(S-\operatorname{id})P_1=1(v_1-2v+v_{-1}),\quad (S-\operatorname{id})P_2=n(v_1-2v+v_{-1}).
\end{equation}
By using the homotopy method (see, e.g. \cite{Hydon2004}) or by inspection, we obtain the functions $P_1$ and $P_2$:
\begin{equation}
P_1=v-v_{-1},\quad P_2=(n-1)v-nv_{-1}.
\end{equation}
This gives the general solution $v=nc_1+c_2$ and hence $u=\exp(nc_1+c_2)$.
\end{exm}

In fact, the above example is related with the following theorem. 
\begin{thm}\label{phydon}
Assume that, by some point transformation, a system of ordinary difference equations can be put into Euler-Lagrange form which admits (Lie point) variational symmetries, then we can conclude that the original system has first integrals that are invariant under Lie point transformations which transform to these variational symmetries.
\end{thm}
\begin{proof}

Consider a system of ordinary difference equations $\{F_{\alpha}(n,[u])=0\}$, where $n\in\mathbb{Z}$. Let us assume that under an (invertible) point transformation $T:\widetilde{u}^{\alpha}_n=T^{\alpha}(n,u_n)$, the new system is
\begin{equation}\label{nds}
0=\widetilde{F}_{\alpha}(n,[\widetilde{u}])=F_{\alpha}(n,[T^{-1}(n,\widetilde{u}_n)]).
\end{equation}
Assume under the new coordinates $(n,[\widetilde{u}])$, the system is governed by a Lagrangian $\widetilde{L}_n=L(n,[\widetilde{u}])$, and there exists a (Lie point) variational symmetry $\widetilde{X}= Q^{\alpha}(n,\widetilde{u})\frac{\partial}{\partial \widetilde{u}^{\alpha}}$ such that
\begin{equation}
\bold{pr}\widetilde{X}(\widetilde{L}_n)=(S-\operatorname{id})P_0.
\end{equation}
Now $\widetilde{X}$ is also a symmetry generator for the new system (\ref{nds}). According to \cite{Peng2014}, there exists some function $\widetilde{P}=P(n,[\widetilde{u}])$ such that the following equality holds:
\begin{equation}
\bold{pr}\widetilde{X}(\widetilde{L}_n)-(S-\operatorname{id})P_0=Q^{\alpha}(n,\widetilde{u}) \widetilde{F}_{\alpha}+(S-\operatorname{id})\widetilde{P}.
\end{equation}
Namely $\widetilde{P}$ is a first integral of the new system (\ref{nds}) through Noether's theorem using the symmetry $\widetilde{X}$. Hence, it is invariant under $\widetilde{X}$ (see Theorem 3.5 of \cite{Peng2014}):
\begin{equation}
\bold{pr}\widetilde{X}(\widetilde{P})=0.
\end{equation}

Let us summarise here that, for the new system (\ref{nds}), there exists a symmetry generator $\widetilde{X}$ and a first integral $\widetilde{P}$ such that 
\begin{equation}\label{eqs}
\left\{
\begin{array}{l}
Q^{\alpha}(n,\widetilde{u})\widetilde{F}_{\alpha}+(S-\operatorname{id})\widetilde{P}=0,\vspace{0.25cm}\\
\bold{pr}\widetilde{X}(\widetilde{P})=0.
\end{array}
\right.
\end{equation}
By using the point transformation $T^{-1}$, the generator becomes
\begin{equation}
X=Q^{\alpha}(n,T(n,u_n))\frac{\partial u_n^{\beta}}{\partial \widetilde{u}_n^{\alpha}}\frac{\partial}{\partial u_n^{\beta}}.
\end{equation}
It is clearly a symmetry generator for the original system since if we write the linearized symmetry condition as
\begin{equation}
\bold{pr}\widetilde{X}(\widetilde{F}_{\alpha})=\sum_{\beta,k\in\mathbb{Z}}A^{\beta}_{\alpha,k}(n,[\widetilde{u}])(S^k\widetilde{F}_{\beta}),
\end{equation}
using $T^{-1}$, it still holds in the original coordinates as follows:
\begin{equation}
\bold{pr}X(F_{\alpha})=\sum_{\beta,k\in\mathbb{Z}}A^{\beta}_{\alpha,k}(n,[T(n,u_n)])(S^kF_{\beta}),
\end{equation}
These equalities in (\ref{eqs}) amount to 
\begin{equation}\label{eq999}
\left\{
\begin{array}{l}
Q^{\alpha}(n,T(n,u_n))F_{\alpha}+(S-\operatorname{id})P(n,[T(n,u_n)])=0,\vspace{0.25cm}\\
\bold{pr}X(P(n,[T(n,u_n)]))=0.
\end{array}
\right.
\end{equation}
 Therefore, we conclude that $P(n,[T(n,u_n)])$ is a first integral for the original system and it is invariant under the Lie point symmetry generator $X$, which can be transformed to a variational symmetry. 
\end{proof}

\begin{rem}
The first equality in (\ref{eq999}) can be rewritten as 
\begin{equation}
\begin{aligned}
0&=Q^{\alpha}(n,T(n,u_n))\frac{\partial u_n^{\beta}}{\partial \widetilde{u}_n^{\alpha}}\frac{\partial \widetilde{u}_n^{\alpha}}{\partial u_n^{\beta}}F_{\alpha}+(S-\operatorname{id})P(n,[T(n,u_n)])\\
&=\left(Q^{\alpha}(n,T(n,u_n))\frac{\partial u_n^{\beta}}{\partial \widetilde{u}_n^{\alpha}}\right)\left(\frac{\partial \widetilde{u}_n^{\gamma}}{\partial u_n^{\beta}}F_{\gamma}\right)+(S-\operatorname{id})P(n,[T(n,u_n)]),
\end{aligned}
\end{equation}
where 
\begin{equation}
Q^{\alpha}(n,T(n,u_n))\frac{\partial u_n^{\beta}}{\partial \widetilde{u}_n^{\alpha}}
\end{equation}
 is the characteristic for $X$. It is obvious that the difference system 
 \begin{equation}
 \left\{\frac{\partial \widetilde{u}_n^{\gamma}}{\partial u_n^{\beta}}F_{\gamma}=0\right\}
 \end{equation}
  is equivalent to the original system $\{F_{\alpha}=0\}$ since the transformation is invertible. This hence provides the relationship between characteristics of symmetries and first integrals.
\end{rem}

\subsection{Adjoint system and conservation laws}

For both differential and difference systems, we summarised above that for systems of the Kovalevskaya form, characteristics of symmetries and conservation laws, respectively, lie in the kernels of the Fr\'echet derivative and its adjoint. The {\bf self-adjointness} or {\bf formal Lagrangian} method by Ibragimov \cite{Ibragimov2007} (see also \cite{Anco1996,Anco1997}) provides a convenient approach for constructing conservation laws from symmetries; in some sense, it is a generalisation of self-adjointness of the Fr\'echet derivative and can be particularly interesting if the system at hand is not variational. It is interesting to notice that formal Lagrangians have a longer history; in literature they appeared as associated (or related) variational principles of differential equations (e.g. \cite{Bateman1931,AH1975}), where researchers paid much attention to inverse problems. Extension of this method to difference systems has also been investigated \cite{Peng2015}. Section 4 will be mainly based on this method, which turns out to be more convenient to the study of some integrable DDEs for the reason that their symmetries are often already well known.

Here, we summarise the algorithm briefly, no matter if the system $\{F_{\alpha}(\cdot,[u])=0\}$ is differential or difference, where $\cdot$ is either $x$ or $n$ . 
\begin{enumerate}
\item Assume the system $\{F_{\alpha}(\cdot,[u])=0\}$ admits a group of symmetries with infinitesimal generators $X$.
\item Introduce new variables $v$ and define a formal Lagrangian $L=v^{\alpha}F_{\alpha}$. At the moment, all symmetries $X$ of the system can be extended to variational symmetries $Y$ of $L$ in the coordinates $(\cdot,u,v)$ (see, \cite{Ibragimov2007} for a differential version and \cite{Peng2015} for a difference version). In Section \ref{section4}, we will show that one can obtain conservation laws without calculating extended variational symmetries $Y$ explicitly (see Remark \ref{withouty}). 
\item Calculate conservation laws $P(\cdot,[u],[v])$ through Noether's theorem.
\item Testify the self-adjointness\footnote{Self-adjointness of an equation is not the same as self-adjointness of the associated Fr\'echet derivative.} of the original system $\{F_{\alpha}=0\}$. We may move to the next step if the original system is self-adjoint, namely there exists a substitution $v=f(\cdot,[u])$ such that the adjoint system 
\begin{equation}
\left\{\frac{\delta L}{\delta u^{\alpha}}=0\right\}
\end{equation}
 holds for all solutions of the original system. Here $\delta u$ is the variation with respect to $u$. 
 \item Finally, one obtains conservation laws $P(\cdot,[u],[f(\cdot,[u])])$ for the original system $\{F_{\alpha}=0\}$.
 \end{enumerate}


It is worth mentioning the fact that for both differential and difference systems, the Lie bracket of two evolutionary symmetry generators is still an evolutionary symmetry generator (see, e.g. \cite{Olver1993} for a differential version and \cite{Peng2015} for a difference version). Namely, for two infinitesimal generators $X_1=Q_1^{\alpha}(\cdot,[u])\partial_{u^{\alpha}}$ and $X_2=Q_2^{\alpha}(\cdot,[u])\partial_{u^{\alpha}}$ of either a differential or difference system, their Lie bracket $[X_1,X_2]$ is still a symmetry generator whose characteristic is
\begin{equation}
Q(\cdot,[u])=\bold{pr}X_1(Q_2)-\bold{pr}X_2(Q_1).
\end{equation}
For a system with infinitely many (evolutionary) symmetries, we are then able to obtain infinitely many (divergence) variational symmetries of the evolutionary form for the formal Lagrangian; their characteristics are then characteristics of conservation laws for the combined system of the original system and its adjoint system. They contribute to conservation laws of the original system through, for instance, the homotopy method, if the system is self-adjoint. Concrete examples can be found in, for instance\cite{Ibragimov2007,Peng2015}. We illustrate the self-adjointness of the discrete KdV equation as a simple example here.

\begin{exm}[Self-adjointness of the discrete KdV equation]
The discrete KdV equation (H1 equation from the ABS list, see e.g. \cite{ABS2003,Hirota1977,Nijhoff1995}) reads 
\begin{equation}\label{dkdv00}
(u_{0,0}-u_{1,1})(u_{1,0}-u_{0,1})-a(m)+b(n)=0.
\end{equation}
Here we write $u_{0,0}=u(m,n)$ at the point $(m,n)\in\mathbb{Z}^2$. Its shifts are $u_{i,j}=S_m^iS_n^ju_{0,0}$ with $S_m$ and $S_n$ the unit forward shifts along the $m$ and $n$ directions, respectively. This system admits a variational structure (e.g. \cite{Capel1991,Lobb2009}), if we rewrite the equation as
\begin{equation}\label{dkdv}
u_{1,0}-u_{0,1}-\frac{a(m)-b(n)}{u_{0,0}-u_{1,1}}=0.
\end{equation}
The difference Euler-Lagrange equation with respect to the Lagrangian 
\begin{equation}
L(m,n,[u])=(u_{1,0}-u_{0,1})u_{0,0}-(a(m)-b(n))\ln(u_{1,0}-u_{0,1})
\end{equation}
contains two copies of the equation (\ref{dkdv}), namely
\begin{equation}
\begin{aligned}
0&=\bold{E}^{\vartriangle}(L(m,n,[u]))\\
&=u_{-1,0}-u_{0,1}-\frac{a(m-1)-b(n)}{u_{0,0}-u_{-1,1}}-\left(u_{0,-1}-u_{1,0}-\frac{a(m)-b(n-1)}{u_{1,-1}-u_{0,0}}\right).
\end{aligned}
\end{equation}


Alternatively, we define a (discrete) formal Lagrangian 
\begin{equation}
L=v_{0,0}\left(u_{1,0}-u_{0,1}-\frac{a(m)-b(n)}{u_{0,0}-u_{1,1}}\right),
\end{equation}
and the adjoint equation is 
\begin{equation}
\begin{aligned}
0&=\bold{E}^{\vartriangle}_{u}(L)\\
&=v_{-1,0}-v_{0,-1}+\frac{a(m)-b(n)}{(u_{0,0}-u_{1,1})^2}v_{0,0}-\frac{a(m-1)-b(n-1)}{(u_{-1,-1}-u_{0,0})^2}v_{-1,-1}.
\end{aligned}
\end{equation}
This becomes two copies of the discrete KdV equation (\ref{dkdv}) via the substitution 
\begin{equation}\label{dkdvsub}
v_{0,0}=(-1)^{m+n}(u_{1,1}-u_{0,0}).
\end{equation}
Now the adjoint equation becomes
\begin{equation}
(-1)^{m+n}\left[u_{1,0}-u_{0,1}-\frac{a(m)-b(n)}{u_{0,0}-u_{1,1}}-\left(u_{0,-1}-u_{-1,0}-\frac{a(m-1)-b(n-1)}{u_{-1,-1}-u_{0,0}}\right)\right]=0,
\end{equation}
which holds for all solutions of the discrete KdV equation. Hence the discrete KdV equation is self-adjoint.

\end{exm}

\section{Noether's theorem for differential-difference equations}
The theory of variational principle (least action) for differential-difference equations (DDEs) was introduced in \cite{Kuper1985} for the most general case, namely for a Lagrangian defined on dependent variables $u\in\mathbb{R}^q$ and finitely many of their derivatives and shifts, with multidimensional differential and difference variables $x\in\mathbb{R}^{p_1}$ and $n\in\mathbb{Z}^{p_2}$ playing as independent variables.  The derivations and applications (e.g. as an integrability criterion or for conducting reductions) of symmetries and conservation laws for DDEs have also been deeply investigated during the last few decades, see for instance \cite{Jiang1998,Li2008,Shen2010,Cherdantsev1995,Yamilov2006,Khanizadeh2014,Levi1993,Levi2010,Quispel1992,Ste-Marie2009} for symmetry analysis, \cite{Goktas1997,Hickman2003,Zhang2003} for derivation of conservation laws and \cite{Cherdantsev1995,Yamilov2006,Khanizadeh2014,Fu2013,Gari2016,Hu1998,Khanizadeh2013,Nijhoff1990} for integrability method.

To avoid the introduction of new notations, we will (often) use the same notations as the differential case unless otherwise explained. It is not difficult to specify the situation from the context.
For the dependent variables, we define derivatives and shifts simultaneously and we adopt the notation 
\begin{equation}
u_{J_1;J_2}=D_{J_1}S_{J_2}u=S_{J_2}D_{J_1}u.
\end{equation} 
Namely, the first subindex indicates derivatives while the second subindex indicates shifts. We still use $[u]$ to denote $u$ and finitely many of its derivatives and shifts. 
Here the {\bf total derivative in the differential-difference sense} is defined as
\begin{equation}
D_i=\partial_{x^i}+\frac{\partial u^{\alpha}}{\partial x^i}\partial_{u^{\alpha}}+\cdots+\sum_{\alpha,J_1,J_2}u^{\alpha}_{J_1+\bold{1}_i;J_2}\partial_{u^{\alpha}_{J_1;J_2}}.
\end{equation}
\begin{defn}
A system of DDEs 
\begin{equation}
\mathcal{A}=\{F_{\alpha}(x,n,[u])=0\}
\end{equation}
is of the {\bf Kovalevskaya form} if it can be rewritten of either the differential Kovalevskaya form
\begin{equation}\label{dKf1}
\frac{\partial^k u^{\alpha}}{\partial (x^1)^k}=f^{\alpha}\left(x,n,[u]_{x^1},\left[\frac{\partial^1 u}{\partial (x^1)^1}\right]_{x^1},\left[\frac{\partial^2 u}{\partial (x^1)^2}\right]_{x^1},\ldots,\left[\frac{\partial^{k-1} u}{\partial (x^1)^{k-1}}\right]_{x^1}\right),
\end{equation}
 or the difference Kovalevskaya form
\begin{equation}\label{dKf2}
S_{n^1}^{k}u^{\alpha}=h^{\alpha}\left(x,n,[u]_{n^1},\left[S_{n^1}^{1}u\right]_{n^1},\left[S_{n^1}^{2}u\right]_{n^1},\ldots,\left[S_{n^1}^{k-1}u\right]_{n^1}\right).
\end{equation}
Here $[\cdot]_{x^1}$ (respectively $[\cdot]_{n^1}$) denotes $\cdot$ itself and finitely many of its derivatives and shifts with respect to all independent variables but $x^1$ (respectively $n^1$).
  \end{defn}
For instance, both $u'=f(u_n,u_{n+1},u_{n+2})$ and $u_{n+1}=f(u',u'')$ are of the Kovalevskaya form where $t\in\mathbb{R}$ and $n\in\mathbb{Z}$. Here we use the notations 
\begin{equation}
u'=\frac{\operatorname{d}\! u(t,n)}{\operatorname{d}\! t},\quad u''=\frac{\operatorname{d}\!^2 u(t,n)}{\operatorname{d}\! t^2}.
\end{equation}
We call a system of DDEs of the {\bf bi-Kovalevskaya form} if it is of the Kovalevskaya form from both differential sense \eqref{dKf1} and difference sense \eqref{dKf2}.

Consider a generalised vector field $X=\xi^i(x,n,[u])\partial_{x^i}+\phi^{\alpha}(x,n,[u])\partial_{u^{\alpha}}$, which generates a one-parameter group of transformations $\widetilde{x}=\widetilde{x}(\varepsilon;x,n,[u])$ and $\widetilde{u}=\widetilde{u}(\varepsilon;x,n,[u])$ and assume the identity element is $\varepsilon=e$. Its {\bf prolongation in the differential-difference sense} is 
\begin{equation}\label{pdde}
\bold{pr}X=\xi^i\partial_{x^i}+\phi^{\alpha}\partial_{u^{\alpha}}+\cdots+\phi^{\alpha}_{J_1;J_2}\partial_{u^{\alpha}_{J_1;J_2}}+\cdots,
\end{equation} 
where
\begin{equation}
\xi^i(x,n,[u]):=\frac{\operatorname{d}\!\widetilde{x}^i}{\operatorname{d}\!\varepsilon}\Big|_{\varepsilon=e}
, \phi^{\alpha}(x,n,[u]):=\frac{\operatorname{d}\! \widetilde{u}^{\alpha}}{\operatorname{d}\!\varepsilon}\Big|_{\varepsilon=e} \quad \text{and} \quad \phi^{\alpha}_{J_1;J_2}(x,n,[u]):=\frac{\operatorname{d}\! \widetilde{u_{J_1;J_2}^{\alpha}}}{\operatorname{d}\!\varepsilon}\Big|_{\varepsilon=e}.
\end{equation}

Next we are going to investigate recursion relations for the coefficient functions of  $\bold{pr}X$. We will  consider first vector fields for Lie point transformations, i.e. $X=\xi^i(x,n)\partial_{x^i}+\phi^{\alpha}(x,n)\partial_{u^{\alpha}}$, and then generalised vector fields.
 It has been pointed out that simple shift is possibly not enough for the recursion relation for difference (or discrete) directions \cite{Levi2010}; in fact, the reason is due to the lack of general commutativity, namely $\widetilde{D}S\neq S\widetilde{D}$ where $\widetilde{D}$ is the total derivative with respect to new variables $\widetilde{x}$. Nevertheless, the prolongation formulae we end up with are still different from the one obtained in \cite{Levi2010}.
We will consider both cases (either $\widetilde{D}S$ or $S\widetilde{D}$) separately (see below). 
They are equal only when $\widetilde{x}$ is independent from $n$ and $[u]$. 

\begin{exm}
Let $x\in\mathbb{R}$ and $n\in\mathbb{Z}$, and let $u(t,n)$ be a scalar dependent variable. Consider the following local transformations
\begin{equation}
\widetilde{t}=\widetilde{t}(\varepsilon;t,n,u),\quad \widetilde{u}=\widetilde{u}(\varepsilon;t,n,u).
\end{equation}
There are two ways to understand, for instance, the notation $\widetilde{u}_{1}'(\widetilde{t},n)$. Do we do the derivative $\widetilde{D}={\operatorname{d}}/{\operatorname{d}\!\widetilde{t}}$ first or the shift $S$ first? Apparently they may lead to different conclusions since in general
\begin{equation}
S\frac{\operatorname{d}\!\widetilde{u}}{\operatorname{d}\!\widetilde{t}}=\frac{\operatorname{d}(S\widetilde{u})}{\operatorname{d}(S\widetilde{t})}\neq \frac{\operatorname{d}(S\widetilde{u})}{\operatorname{d}\!\widetilde{t}}.
\end{equation}
\end{exm}

Besides the total derivatives $D_i$, we also define {\bf multi-total derivatives in the differential-difference sense} for a given index $I$ as
\begin{equation}
D_{i;I}=\partial_{x^i}+\frac{\partial u_{\bold{0};I}^{\alpha}}{\partial x^i}\partial_{u_{\bold{0};I}^{\alpha}}+\cdots+\sum_{\alpha,J_1}u^{\alpha}_{J_1+\bold{1}_i;I}\partial_{u^{\alpha}_{J_1;I}}.
\end{equation}
They are related by 
\begin{equation}
D_i=\partial_{x^i}+\sum_{I}\left(D_{i;I}-\partial_{x^i}\right).
\end{equation}

Next the prolongation formulae of a generalised vector field $X=\xi^i(x,n,[u])\partial_{x^i}+\phi^{\alpha}(x,n,[u])\partial_{u^{\alpha}}$ are derived for the cases $\widetilde{D}S$ (i.e. all the shifts followed by all the derivatives) and $S\widetilde{D}$ (i.e. all the derivatives followed by all shifts) separately.

At the moment, let us first consider Lie point transformations, that is, $\xi^i=\xi^i(x,n,u)$ and $\phi^{\alpha}=\phi^{\alpha}(x,n,u)$.

{\bf Case I:  $\widetilde{D}S$.} 
Now the prolongations of transformations are given by
\begin{equation}\label{pot1}
\widetilde{u_{J_1;J_2}^{\alpha}}:=\widetilde{D_{J_1}}S_{J_2}\widetilde{u}^{\alpha}.
\end{equation}

\begin{prop}\label{prop33}
We have the following recursion relations
\begin{equation}\label{dvrr}
\phi^{\alpha}_{J_1+\bold{1}_i;J_2}=D_i\phi^{\alpha}_{J_1;J_2}-(D_i\xi^j)u^{\alpha}_{J_1+\bold{1}_j;J_2}
\end{equation}
and
\begin{equation}\label{prsk}
\phi^{\alpha}_{\bold{0};J_2+\bold{1}_k}=S_k\phi^{\alpha}_{\bold{0;}J_2}
\end{equation}
for all valid indices $\alpha,i,j,k,J_1,J_2$.
\end{prop}

\begin{proof}
The equality (\ref{dvrr}) is a consequence of the fact that (for any matrix $M(\varepsilon)$)
\begin{equation}
D_{\widetilde{x}}\widetilde{u^{\alpha}_{J_1;J_2}}=\left(D_x\widetilde{u^{\alpha}_{J_1;J_2}}\right) (D_x\widetilde{x})^{-1} \text{ and } \frac{\operatorname{d}}{\operatorname{d}\!\varepsilon}(M(\varepsilon)^{-1})=-M(\varepsilon)^{-1}\frac{\operatorname{d}\!M(\varepsilon)}{\operatorname{d}\!\varepsilon}M(\varepsilon)^{-1}.
\end{equation}
Hence for any $j$,
\begin{equation}
\frac{\operatorname{d}\! \widetilde{u^{\alpha}_{J_1+\bold{1}_i;J_2}}}{\operatorname{d}\! \varepsilon}\Big|_{\varepsilon=e}=D_i\frac{\operatorname{d}\! \widetilde{u^{\alpha}_{J_1;J_2}}}{\operatorname{d}\!\varepsilon}\Big|_{\varepsilon=e}-\left(D_i\frac{\operatorname{d}\! \widetilde{x}^j}{\operatorname{d}\! \varepsilon}\Big|_{\varepsilon=e}\right)u_{J_1+\bold{1}_j;J_2}^{\alpha},
\end{equation}
which is exactly (\ref{dvrr}).

The second part is due to the fact that 
\begin{equation}
S_k\widetilde{u^{\alpha}_{\bold{0};J_2}}=\widetilde{u^{\alpha}_{\bold{0};J_2+\bold{1}_k}}.
\end{equation}
\end{proof}

\begin{prop}\label{prop34}
For any given indices $J_1,J_2$, the following equality holds 
\begin{equation}
\phi^{\alpha}_{J_1;J_2+\bold{1}_k}=S_k\phi^{\alpha}_{J_1;J_2}+D_{J_1}\left(\left(S_k\xi^i-\xi^i\right)u^{\alpha}_{\bold{1}_i;J_2+\bold{1}_k}\right)-\left(S_k\xi^i-\xi^i\right)u^{\alpha}_{J_1+\bold{1}_i;J_2+\bold{1}_k}.
\end{equation}
In particular when $J_1=\bold{0}$, this agrees with the result in Proposition \ref{prop33}.
\end{prop}

\begin{proof}
The recursion relations (\ref{dvrr}) in Proposition \ref{prop33} are exactly the same as those in the differential case for all given indices $J_2$. Then we have that 
\begin{equation}\label{dvrr1}
\phi^{\alpha}_{J_1;J_2}=D_{J_1}\left(\phi^{\alpha}_{\bold{0};J_2}-\xi^iu^{\alpha}_{\bold{1}_i;J_2}\right)+\xi^iu^{\alpha}_{J_1+\bold{1}_i;J_2}
\end{equation}
and 
\begin{equation}
\phi^{\alpha}_{J_1;J_2+\bold{1}_k}=D_{J_1}\left(\phi^{\alpha}_{\bold{0};J_2+\bold{1}_k}-\xi^iu^{\alpha}_{\bold{1}_i;J_2+\bold{1}_k}\right)+\xi^iu^{\alpha}_{J_1+\bold{1}_i;J_2+\bold{1}_k}.
\end{equation}
Immediately we obtain 
\begin{equation}
\begin{aligned}
\phi^{\alpha}_{J_1;J_2+\bold{1}_k}-S_k\phi^{\alpha}_{J_1;J_2}=&~D_{J_1}\left(\phi^{\alpha}_{\bold{0};J_2+\bold{1}_k}-\xi^iu^{\alpha}_{\bold{1}_i;J_2+\bold{1}_k}\right)+\xi^iu^{\alpha}_{J_1+\bold{1}_i;J_2+\bold{1}_k}\\
&~-S_kD_{J_1}\left(\phi^{\alpha}_{\bold{0};J_2}-\xi^iu^{\alpha}_{\bold{1}_i;J_2}\right)-(S_k\xi^i)u^{\alpha}_{J_1+\bold{1}_i;J_2+\bold{1}_k}\\
=&~D_{J_1}\left(\left(S_k\xi^i-\xi^i\right)u^{\alpha}_{\bold{1}_i;J_2+\bold{1}_k}\right)-\left(S_k\xi^i-\xi^i\right)u^{\alpha}_{J_1+\bold{1}_i;J_2+\bold{1}_k}.
\end{aligned}
\end{equation}
\end{proof}

\begin{thm}\label{thm33}
In this case, the prolongation formula (\ref{pdde}) is then 
\begin{equation}
\bold{pr}X=\xi^iD_i+\sum_{\alpha,J_1,J_2}\left(D_{J_1}Q_{J_2}^{\alpha}\right)\partial_{u^{\alpha}_{J_1;J_2}},
\end{equation}
where the {\bf multi-characteristics} are defined by
\begin{equation}
Q^{\alpha}_{J_2}:=S_{J_2}\phi^{\alpha}-\xi^iu^{\alpha}_{\bold{1}_i;J_2}.
\end{equation}
\end{thm}
\begin{proof}
By substituting (\ref{dvrr1}) in the prolongation formula, direct calculation gives
\begin{equation}
\begin{aligned}
\bold{pr}X=&~\xi^i\partial_{x^i}+\phi^{\alpha}\partial_{u^{\alpha}}+\cdots+(S_{J_2}\phi^{\alpha})\partial_{u^{\alpha}_{\bold{0};J_2}}+\cdots\\
&~+\cdots\\
&~+\left(D_{J_1}Q^{\alpha}+\xi^iu^{\alpha}_{J_1+\bold{1}_i;\bold{0}}\right)\partial_{u^{\alpha}_{J_1;\bold{0}}}+\cdots+\left(D_{J_1}Q^{\alpha}_{J_2}+\xi^iu^{\alpha}_{J_1+\bold{1}_i;J_2}\right)\partial_{u^{\alpha}_{J_1;J_2}}+\cdots\\
&~+\cdots\\
=&~\xi^iD_i+Q^{\alpha}\partial_{u^{\alpha}}+\cdots+Q^{\alpha}_{J_2}\partial_{u^{\alpha}_{\bold{0}:J_2}}+\cdots\\
&~+\cdots\\
&~+\left(D_{J_1}Q^{\alpha}\right)\partial_{u^{\alpha}_{J_1;\bold{0}}}+\cdots+\left(D_{J_1}Q^{\alpha}_{J_2}\right)\partial_{u^{\alpha}_{J_1;J_2}}+\cdots\\
&~+\cdots\\
=&~\xi^iD_i+\sum_{\alpha,J_1,J_2}\left(D_{J_1}Q_{J_2}^{\alpha}\right)\partial_{u^{\alpha}_{J_1;J_2}}.
\end{aligned}
\end{equation}
\end{proof}

From Theorem \ref{thm33}, we conclude that the prolonged vector field is equivalent to one in multi-characteristic form using the multi-characteristics $Q^{\alpha}_{J_2}$, namely
\begin{equation}
\sum_{\alpha,J_1,J_2}\left(D_{J_1}Q_{J_2}^{\alpha}\right)\partial_{u^{\alpha}_{J_1;J_2}}.
\end{equation}
 However, it is not yet evolutionary as $Q_{J_2+\bold{1}_k}^{\alpha}\neq S_kQ_{J_2}^{\alpha}$ in general. Prolongation of a vector field is equivalent to an evolutionary vector field only when $\xi=\xi(x)$; in this situation
\begin{equation}
\begin{aligned}
\bold{pr}X
&=\xi^iD_i+\sum_{\alpha,J_1,J_2}\left(D_{J_1}Q^{\alpha}_{J_2}\right)\partial_{u_{J_1;J_2}^{\alpha}}\\
&=\xi^iD_i+\sum_{\alpha,J_1,J_2}\left(D_{J_1}S_{J_2}Q^{\alpha}\right)\partial_{u_{J_1;J_2}^{\alpha}},
\end{aligned}
\end{equation}
where $Q^{\alpha}=\phi^{\alpha}-\xi^iu^{\alpha}_{\bold{1}_i;\bold{0}}$.

{\bf Case II: $S\widetilde{D}$.} In this case, the prolongations of transformations are given by
\begin{equation}\label{pot2}
\widetilde{u_{J_1;J_2}^{\alpha}}:=S_{J_2}\widetilde{D_{J_1}}\widetilde{u}^{\alpha}.
\end{equation}

\begin{prop}\label{prop31}
We have the following recursion relations 
\begin{equation}\label{dvrrcase2}
\phi^{\alpha}_{J_1+\bold{1}_i;J_2}=D_i\phi^{\alpha}_{J_1;J_2}-\left(S_{J_2}D_i\xi^j\right)u_{J_1+\bold{1}_j;J_2}^{\alpha}
\end{equation}
and
\begin{equation}\label{eqsk108}
S_k\phi_{J_1;J_2}^{\alpha}=\phi_{J_1;J_2+\bold{1}_k}^{\alpha}
\end{equation}
for all valid indices $\alpha,i,j,k,J_1,J_2$.
\end{prop}

\begin{proof}
First we have
\begin{equation}
\begin{aligned}
\widetilde{u^{\alpha}_{J_1+\bold{1};J_2}}&=S_{J_2}\left(D_{\widetilde{x}}\widetilde{u^{\alpha}_{J_1;\bold{0}}}\right)=(S_{J_2}D_{\widetilde{x}})\widetilde{u^{\alpha}_{J_1;J_2}}\\
&=\left(D_x\widetilde{u^{\alpha}_{J_1;J_2}}\right)S_{J_2}(D_x\widetilde{x})^{-1}.
\end{aligned}
\end{equation}
The derivative with respect to $\varepsilon$ at $e$ on both sides yields
\begin{equation}
\phi^{\alpha}_{J_1+\bold{1}_i;J_2}=D_i\phi^{\alpha}_{J_1;J_2}-\left(S_{J_2}D_i\xi^j\right)u_{J_1+\bold{1}_j;J_2}^{\alpha},
\end{equation}
which finishes the proof of the first piece of results.
The equality (\ref{eqsk108}) can be immediately proved by noting that
\begin{equation}
S_k\widetilde{u^{\alpha}_{J_1;J_2}}=S_k\left(S_{J_2}\left(\widetilde{D_{J_1}}\widetilde{u}^{\alpha}\right)\right)=S_{J_2+\bold{1}_k}\left(\widetilde{D_{J_1}}\widetilde{u}^{\alpha}\right)=\widetilde{u^{\alpha}_{J_1;J_2+\bold{1}_k}}.
\end{equation}
\end{proof}

\begin{thm}\label{thmcase1}
The prolongation formula (\ref{pdde}) can be expressed in terms of its {\bf characteristic} $Q^{\alpha}=\phi^{\alpha}-\xi^iu^{\alpha}_{\bold{1}_i;\bold{0}}$ as
\begin{equation}\label{eq33}
\bold{pr}X=\xi^i\partial_{x^i}+\sum_{i,I}(S_I\xi^i)\left(D_{i;I}-\partial_{x^i}\right)+\sum_{\alpha,J_1,J_2}(D_{J_1}S_{J_2}Q^{\alpha})\partial_{u^{\alpha}_{J_1;J_2}}.
\end{equation}
\end{thm}

\begin{proof}
The recursion relations in Proposition \ref{prop31} lead to that
\begin{equation}
\begin{aligned}
\phi^{\alpha}_{J_1;J_2}&=D_{J_1}\left(\phi^{\alpha}_{\bold{0};J_2}-\left(S_{J_2}\xi^i\right)u^{\alpha}_{\bold{1}_i;J_2}\right)+\left(S_{J_2}\xi^i\right)u^{\alpha}_{J_1+\bold{1}_i;J_2}\\
&=D_{J_1}S_{J_2}Q^{\alpha}+\left(S_{J_2}\xi^i\right)u^{\alpha}_{J_1+\bold{1}_i;J_2}\\
&=S_{J_2}\left(D_{J_1}Q^{\alpha}+\xi^iu^{\alpha}_{J_1+\bold{1}_i;\bold{0}}\right).
\end{aligned}
\end{equation}
Substituting these back to the prolongation formula we get 
\begin{equation}
\begin{aligned}
\bold{pr}X=&~\xi^i\partial_{x^i}+\left(Q^{\alpha}+\xi^iu^{\alpha}_{\bold{1}_i;\bold{0}}\right)\partial_{u^{\alpha}}+\cdots+S_{J_2}\left(Q^{\alpha}+\xi^iu^{\alpha}_{\bold{1}_i;\bold{0}}\right)\partial_{u^{\alpha}_{\bold{0};J_2}}+\cdots\\
&+\cdots\\
&+\left(D_{J_1}Q^{\alpha}+\xi^iu_{J_1+\bold{1}_i;\bold{0}}^{\alpha}\right)\partial_{u^{\alpha}_{J_1;\bold{0}}}+\cdots+S_{J_2}\left(D_{J_1}Q^{\alpha}+\xi^iu_{J_1+\bold{1}_i;\bold{0}}^{\alpha}\right)\partial_{u^{\alpha}_{J_1;J_2}}+\cdots\\
&+\cdots\\
=&~ \xi^i\partial_{x^i}+\xi^i\left(u^{\alpha}_{\bold{1}_i;\bold{0}}\partial_{u^{\alpha}}+\cdots+u^{\alpha}_{J_1+\bold{1}_i;\bold{0}}\partial_{u^{\alpha}_{J_1;\bold{0}}}+\cdots\right)\\
&+\cdots\\
&+ \left(S_{J_2}\xi^i\right)\left(u^{\alpha}_{\bold{1}_i;J_2}\partial_{u^{\alpha}_{\bold{0};J_2}}+\cdots+u^{\alpha}_{J_1+\bold{1}_i;J_2}\partial_{u^{\alpha}_{J_1;J_2}}+\cdots\right)\\
&+\cdots\\
&+ \sum_{\alpha,J_1,J_2}(D_{J_1}S_{J_2}Q^{\alpha})\partial_{u^{\alpha}_{J_1;J_2}}\\
=&~\xi^i\partial_{x^i}+\sum_{i,I}(S_I\xi^i)D_{i;I}-\sum_{i,I}(S_I\xi^i)\partial_{x^i}+\sum_{\alpha,J_1,J_2}(D_{J_1}S_{J_2}Q^{\alpha})\partial_{u^{\alpha}_{J_1;J_2}}.
\end{aligned}
\end{equation}
This finishes the proof.
\end{proof}

From Theorem \ref{thmcase1}, it is clear that in this case:
\begin{itemize}
\item If $\xi$ depend on $(n,[u])$, it is often not possible to write the prolonged vector field equivalently in the evolutionary form since the first two terms on the right hand side of \eqref{eq33} do not contribute to a total derivative term.
\item If $\xi=\xi(x)$, the prolongation can be equivalently written in the evolutionary form because now we have
\begin{equation}\label{eqhaha127}
\begin{aligned}
\bold{pr}X&=\xi^i\left(\partial_{x^i}+\sum_I\left(D_{i;I}-\partial_{x^i}\right)\right)+\sum_{\alpha,J_1,J_2}(D_{J_1}S_{J_2}Q^{\alpha})\partial_{u^{\alpha}_{J_1;J_2}}\\
&=\xi^iD_i+\sum_{\alpha,J_1,J_2}\left(D_{J_1}S_{J_2}Q^{\alpha}\right)\partial_{u_{J_1;J_2}^{\alpha}},
\end{aligned}
\end{equation}
where $Q^{\alpha}=\phi^{\alpha}-\xi^iu^{\alpha}_{\bold{1}_i,\bold{0}}$.
\end{itemize}

\begin{rem}
From the above analysis, we conclude that, for both Case I: $\widetilde{D}S$ and Case II: $S\widetilde{D}$, their prolonged vector fields are of the same form and are equivalent to an evolutionary vector field only when $\xi=\xi(x)$. It is interesting to note that there are also intermediate cases, for instance, $S\widetilde{D}S\widetilde{D}$ or $S\widetilde{D}\widetilde{D}S$. However, we do not attempt to be exhaustive in this paper.
\end{rem}

For a generalised vector field $X=\xi^i(x,n,[u])\partial_{x^i}+\phi^{\alpha}(x,n,[u])\partial_{u^{\alpha}}$, its prolongation $\bold{pr}X$ can then be defined in accordance with either Theorem \ref{thm33} for {\bf Case I:} $\widetilde{D}S$ or Theorem \ref{thmcase1} for {\bf Case II:} $S\widetilde{D}$. When $\xi=\xi(x)$, the prolongation is given uniquely in accordance with \eqref{eqhaha127}.

\begin{defn}
A local group of transformations is called a {\bf symmetry group} for a system of  DDEs $\mathcal{A}=\{F_{\alpha}(x,n,[u])=0\}$ if for all the corresponding generalised infinitesimal generators $X=\xi^i(x,n,[u])\partial_{x^i}+\phi^{\alpha}(x,n,[u])\partial_{u^{\alpha}}$, the following {\bf linearized symmetry condition} is satisfied:
\begin{equation}\label{lscddes}
\bold{pr}X(F_{\alpha})=0, \text{ whenever } \mathcal{A} \text{ holds}.
\end{equation}
Here the prolongation formula is defined in accordance with either Theorem \ref{thm33} or Theorem \ref{thmcase1}.
We say that $X$ generates a group of {\bf Lie point symmetries} if the coefficients $\xi^i$ and $\phi^{\alpha}$ depend only on $x,n$ and $u$.
\end{defn}

It is clear that for DDEs $\mathcal{A}=\{F_{\alpha}(x,n,[u])=0\}$ of the form
\begin{equation}
F_{\alpha}(x,n,[u])=F_{\alpha}\left(x,n,u^{\beta}_{\bold{1}_i;\bold{0}},\ldots,u^{\beta}_{J_1;\bold{0}},u^{\beta}_{\bold{0},\bold{1}_i},u^{\beta}_{\bold{0};-\bold{1}_i},\ldots,u^{\beta}_{\bold{0};J_2}\right)
\end{equation}
for some finite indices $J_1,J_2$, both cases ($\widetilde{D}S$ and $S\widetilde{D}$) are valid for the calculation of the prolongation formula of an infinitesimal generator. The reason is that no term with a mixture of derivatives and shifts appear in $\mathcal{A}$. Particular examples are these of the form
\begin{equation}
u'=f(x,n,u_{-l_1},u_{-l_1+1},\ldots,u,\ldots,u_{l_2}),
\end{equation}
for some integers $l_1,l_2$.
Yamilov's classification of integrable DDEs $u'=f(u_1,u,u_1)$ is clearly in this form \cite{Yamilov2006}. 

For a system of DDEs of the bi-Kovalevskaya form,  the linearized symmetry condition is then equivalent to the existence of functions $K^{\beta}_{\alpha;J_1,J_2}(x,n,[u])$, such that the following identity holds
\begin{equation}
\bold{pr}X(F_{\alpha})=\sum_{\beta,J_1,J_2}K^{\beta}_{\alpha;J_1,J_2}(D_{J_1}S_{J_2}F_{\beta}).
\end{equation}

 In the next running example, the Volterra equation, we show how symmetries can be calculated from the linearized symmetry condition.

\begin{exm}[Volterra equation]\label{volterra}
Let us consider Lie point symmetries for the Volterra equation (e.g. \cite{Yamilov2006,Khanizadeh2013})
\begin{equation}\label{vvveq}
u'-u(u_1-u_{-1})=0,
\end{equation}
which is one of the most simple integrable Volterra type of equations of the form
\begin{equation}
u'=f(u_{-1},u,u_1).
\end{equation}
The equation (\ref{vvveq}) has also been called the KvM lattice (e.g. \cite{Hickman2003,Kac1975}).
Here $t\in\mathbb{R}$ is the differential independent variable while $n\in\mathbb{Z}$ is the difference independent variable. 
 For the purpose of convenience (which will be clear in Section 4), we rewrite the equation as
\begin{equation}
\frac{u'}{u}-u_1+u_{-1}=0.
\end{equation}
Consider a Lie point symmetry generator $X=\xi(t,n,u)\partial_t+\phi(t,n,u)\partial_u$. Its prolongation is 
\begin{equation}
\bold{pr}X=\xi\partial_t+\phi\partial_u+(S\phi)\partial_{u_1}+\left(S_{-1}\phi\right)\partial_{u_{-1}}+\left(D_t\phi-\left(D_t\xi\right)u'\right)\partial_{u'}+\cdots,
\end{equation}
whose first several terms are the same according to either Theorem \ref{thm33} or Theorem \ref{thmcase1}.
From the linearized symmetry condition (\ref{lscddes}), after direct calculations, we have
\begin{equation}\label{vel1}
\begin{aligned}
0&=\bold{pr}X\left(\frac{u'}{u}-u_1+u_{-1}\right)\Big|_{\frac{u'}{u}-u_1+u_{-1}=0}\\
&=-\frac{\phi}{u}(u_1-u_{-1})+\frac{\phi_t}{u}+\phi_u(u_1-u_{-1})-(\xi_t+\xi_uu(u_1-u_{-1}))(u_1-u_{-1})-S\phi+S_{-1}\phi.
\end{aligned}
\end{equation}
By differentiating (\ref{vel1}) consecutively with respect to $u_1$ and $u_{-1}$, we obtain  
\begin{equation}
2\xi_uu=0.
\end{equation}
Hence we have
\begin{equation}
\xi=\xi(t,n),
\end{equation}
and the linearized symmetry condition (\ref{vel1}) simplifies to
\begin{equation}\label{vel2}
0=-\frac{\phi}{u}(u_1-u_{-1})+\frac{\phi_t}{u}+\phi_u(u_1-u_{-1})-\xi_t(u_1-u_{-1})-S\phi+S_{-1}\phi.
\end{equation}
Differentiate (\ref{vel2}) with respect to $u_1$ (or $u_{-1}$) twice gives that 
\begin{equation}
S(\phi_{uu})=0,
\end{equation}
which can be solved, namely
\begin{equation}
\phi=a(t,n)u+b(t,n).
\end{equation}
The resulting $\xi$ and $\phi$ are then substituted back to the linearized symmetry condition (\ref{vel2}); this leads to 
\begin{equation}
0=\frac{b'-b(u_1-u_{-1})}{u}-\left[\xi_t+a(t,n+1)\right]u_1+\left[\xi_t+a(t,n-1)\right]u_{-1}+a'-b(t,n+1)+b(t,n-1).
\end{equation}
The following equations are then achieved
\begin{equation}
\begin{aligned}
b=0,\quad \xi_t+a(t,n+1)=0,\quad a(t,n+1)=a(t,n-1),\quad a'=0,
\end{aligned}
\end{equation}
whose solution is
\begin{equation}
a=c_1+c_2(-1)^n,\quad b=0,\quad \xi=\left(-c_1+c_2(-1)^n\right)t+c_3(n).
\end{equation}
Hence
\begin{equation}
\phi=\left(c_1+c_2(-1)^n\right)u.
\end{equation}
Here $c_1,c_2$ are constants; $c_3(n)$ is an arbitrary function of $n$. Therefore, we obtained all Lie point symmetries for the Volterra equation, namely
\begin{equation}
X_1=-t\partial_t+u\partial_u,\quad X_2=(-1)^nt\partial_t+(-1)^nu\partial_u,\quad X_3=c_3(n)\partial_t.
\end{equation}
\end{exm}

\begin{lem}\label{lemevo}
Given a vector field
\begin{equation}
Y=\xi^i(x,n,[u])\partial_{x^i}+\sum_{\alpha,J_1,J_2}\phi^{\alpha}_{J_1;J_2}(x,n,[u])\partial_{u^{\alpha}_{J_1;J_2}},
\end{equation}
we have $[Y,D_i]=0$ for $i=1,2,\ldots,p_1$ and $[Y,S_j]$ for $j=1,2,\ldots,p_2$ if and only if 
\begin{equation}\label{yevo}
Y=\bold{pr}X+c^i\partial_{x^i},
\end{equation}
where $\bold{pr}X$ is an evolutionary vector field and $c^i$ are constants.
\end{lem}
\begin{proof}
We first perform $Y$ in the form (\ref{yevo}). Then $[Y,D_i]=0$ immediately follows from a similar derivation as Olver did in \cite{Olver1993} (Lemma 5.12). For a given function $f(x,n,[u])$, direct calculation yields
\begin{equation}
\begin{aligned}
\ [Y,S_j]f&=\left(\bold{pr}X+c^i\partial_{x^i}\right)(S_jf)-S_j\left(\bold{pr}X(f)+c^i\frac{\partial f}{\partial x^i}\right)\\
&= \bold{pr}X(S_jf)-S_j(\bold{pr}X(f))\\
&= \sum_{\alpha,J_1,J_2}\left(S_{J_2}D_{J_1}Q^{\alpha}\right)\frac{\partial (S_jf)}{\partial u^{\alpha}_{J_1;J_2}}-S_j\left(\sum_{\alpha,J_1',J_2'}\left(S_{J_2'}D_{J_1'}Q^{\alpha}\right)\frac{\partial f}{\partial u^{\alpha}_{J_1';J_2'}}\right)\\
&= \sum_{\alpha,J_1,J_2}\left(S_{J_2}D_{J_1}Q^{\alpha}\right)\frac{\partial (S_jf)}{\partial u^{\alpha}_{J_1;J_2}}-\sum_{\alpha,J_1',J_2'}\left(S_{J_2'+\bold{1}_j}D_{J_1'}Q^{\alpha}\right)\frac{\partial (S_jf)}{\partial u^{\alpha}_{J_1';J_2'+\bold{1}_j}}\\
&=0
\end{aligned}
\end{equation}
via changes of indices $J_1'=J_1, J_2'=J_2-\bold{1}_j$. 

Conversely, if $[Y,D_i]=0$ holds for $i=1,2,\ldots,p_1$, then we have, similarly to \cite{Olver1993}, that
\begin{equation}
D_i\xi^j=0,\quad \phi^{\alpha}_{J_1+\bold{1}_i;J_2}=D_i\phi^{\alpha}_{J_1;J_2}
\end{equation}
for all indices $\alpha,i,j,J_1,J_2$. Further if $[Y,S_j]=0$ holds for $j=1,2,\ldots,p_2$, we obtain
\begin{equation}
\begin{aligned}
0&=[Y,S_j]f\\
&=\left(\xi^i\partial_{x^i}+\sum_{\alpha,J_1,J_2}\phi^{\alpha}_{J_1;J_2}\partial_{u^{\alpha}_{J_1;J_2}}\right)(S_jf)-S_j\left(\xi^i\frac{\partial f}{\partial x^i}+\sum_{\alpha,J_1',J_2'}\phi^{\alpha}_{J_1';J_2'}\frac{\partial f}{\partial u^{\alpha}_{J_1';J_2'}}\right)\\
&=(\xi^i-S_j\xi^i)\frac{\partial (S_jf)}{\partial x^i}+\sum_{\alpha,J_1,J_2}\left(\phi^{\alpha}_{J_1;J_2}-S_j\phi^{\alpha}_{J_1;J_2-\bold{1}_j}\right)\frac{\partial (S_jf)}{\partial u^{\alpha}_{J_1;J_2}},
\end{aligned}
\end{equation}
where the same changes of indices $J_1'=J_1, J_2'=J_2-\bold{1}_j$ are used. Therefore, we have
\begin{equation}
\xi^i-S_j\xi^i=0,\quad \phi^{\alpha}_{J_1;J_2}=S_j\phi^{\alpha}_{J_1;J_2-\bold{1}_j}
\end{equation}
for all indices $\alpha,i,j,J_1,J_2$. We can conclude now that $\xi^i=c^i$ are constants and 
\begin{equation}
\phi^{\alpha}_{J_1;J_2}=D_{J_1}S_{J_2}\phi^{\alpha}
\end{equation}
for all indices $\alpha,J_1,J_2$. This finishes the proof.
\end{proof}

\begin{thm}
For two evolutionary symmetry generators $X_i=Q^{\alpha}_i(x,n,[u])\partial_{u^{\alpha}}$ ($i=1,2$) for a system of DDEs, their Lie bracket $[X_1,X_2]$ is still a symmetry generator of the evolutionary form where
\begin{equation}
[X_1,X_2]:=\left(\bold{pr}X_1(Q_2^{\alpha})-\bold{pr}X_2(Q_1^{\alpha})\right)\partial_{u^{\alpha}}.
\end{equation}
Its prolongation is 
\begin{equation}
\bold{pr}[X_1,X_2]=[\bold{pr}X_1,\bold{pr}X_2].
\end{equation}
\end{thm}

\begin{proof}
This is a differential-difference counterpart to the conclusions of the differential version (e.g. \cite{Olver1993}) and the difference version (e.g. \cite{Peng2015}). 

We only need prove that $[\bold{pr}X_1,\bold{pr}X_2]$ is of the evolutionary form and its characteristic is the same as that of $\bold{pr}[X_1,X_2]$. First we have that 
\begin{equation}
\begin{aligned}
\ [[\bold{pr}X_1,&\bold{pr}X_2],D_i]=[\bold{pr}X_1\bold{pr}X_2-\bold{pr}X_2\bold{pr}X_1,D_i]\\
&=\left(\bold{pr}X_1\bold{pr}X_2\right)\cdot D_i-\left(\bold{pr}X_2\bold{pr}X_1\right)\cdot D_i-D_i\cdot\left(\bold{pr}X_1\bold{pr}X_2\right)+D_i\cdot \left(\bold{pr}X_2\bold{pr}X_1\right)\\
&=\bold{pr}X_1\left(D_i\cdot \bold{pr}X_2\right)-\bold{pr}X_2\left(\bold{pr}X_1\cdot D_i\right)-\bold{pr}X_1\left(D_i\cdot\bold{pr}X_2\right)+\bold{pr}X_2\left(D_i\cdot \bold{pr}X_1\right)\\
&=-\bold{pr}X_2[\bold{pr}X_1,D_i]\\
&=0
\end{aligned}
\end{equation}
and 
\begin{equation}
\begin{aligned}
\ [[\bold{pr}X_1,&\bold{pr}X_2],S_j]=[\bold{pr}X_1\bold{pr}X_2-\bold{pr}X_2\bold{pr}X_1,S_j]\\
&=\left(\bold{pr}X_1\bold{pr}X_2\right)\cdot S_j-\left(\bold{pr}X_2\bold{pr}X_1\right)\cdot S_j-S_j\cdot\left(\bold{pr}X_1\bold{pr}X_2\right)+S_j\cdot \left(\bold{pr}X_2\bold{pr}X_1\right)\\
&=\bold{pr}X_1\left(S_j\cdot \bold{pr}X_2\right)-\bold{pr}X_2\left(\bold{pr}X_1\cdot S_j\right)-\bold{pr}X_1\left(S_j\cdot\bold{pr}X_2\right)+\bold{pr}X_2\left(S_j\cdot \bold{pr}X_1\right)\\
&=-\bold{pr}X_2[\bold{pr}X_1,S_j]\\
&=0
\end{aligned}
\end{equation}
Here we used the properties that $[\bold{pr}X_k,D_i]=0$ and $[\bold{pr}X_k,S_j]=0$ for all $i=1,2,\ldots,p_1$, $j=1,2,\ldots,p_2$ and $k=1,2$, see Lemma \ref{lemevo} . Therefore, we conclude that $[\bold{pr}X_1,\bold{pr}X_2]$ must be of the evolutionary form since both $\bold{pr}X_1$ and $\bold{pr}X_2$ do not have $\partial_{x^i}$ components. Characteristic of $[\bold{pr}X_1,\bold{pr}X_2]$ can be directly calculated and it is indeed 
\begin{equation}
\bold{pr}X_1(Q_2^{\alpha})-\bold{pr}X_2(Q_1^{\alpha}).
\end{equation}
This finishes the proof.
\end{proof}

As we concluded above that only prolongations of vector fields of the form 
\begin{equation}
X=\xi^i(x)\partial_{x^i}+\phi^{\alpha}(x,n,[u])\partial_{u^{\alpha}}
\end{equation}
 can be equivalently written as the evolutionary form; such type of vector fields deserves a special name and we call them {\bf regular vector fields} or {\bf regular infinitesimal generators} as they agree with the prolongation property of both differential and difference cases. Symmetries generated by regular vector fields will then be called {\bf regular symmetries} for DDEs. For instance, of all the Lie point symmetries of the Volterra equation in Example \ref{volterra}, the vector fields $X_1=-t\partial_t+u\partial_u$ and $X_3=\partial_t$ are regular infinitesimal generators.

From now on, we will only be focused on regular vector fields unless otherwise specified, namely
\begin{equation}\label{proppro}
\bold{pr}X=\xi^i(x)D_i+\sum_{\alpha,J_1,J_2}(D_{J_1}S_{J_2}Q^{\alpha}(x,n,[u]))\partial_{u^{\alpha}_{J_1;J_2}}.
\end{equation}

\begin{rem}\label{evvf}
From the prolongation formula (\ref{proppro}), it is clear that any regular symmetry generator $X=\xi^i(x)\partial_{x^i}+\phi^{\alpha}(x,n,[u])\partial_{u^{\alpha}}$ is equivalent to an evolutionary symmetry generator  $Q^{\alpha}(x,n,[u])\partial_{u^{\alpha}}$, where
\begin{equation}
Q^{\alpha}=\phi^{\alpha}-\xi^iu^{\alpha}_{\bold{1}_i;\bold{0}}.
\end{equation}

\end{rem}

\begin{defn}
For two tuples of functions $A(x,n,[u])$ and $B(x,n,[u])$, we define the {\bf differential-difference Fr\'echet derivative} as
\begin{equation}\label{ddfd}
\bold{D}_A(B):=\frac{\operatorname{d}}{\operatorname{d}\!\varepsilon}\Big|_{\varepsilon=0}A(x,n,[u+\varepsilon B(x,n,[u])]).
\end{equation}
Locally it reads
\begin{equation}
(\bold{D}_A)_{\alpha\beta}=\sum_{J_1,J_2}\frac{\partial A^{\alpha}}{\partial u_{J_1;J_2}^{\beta}}D_{J_1}S_{J_2}.
\end{equation}
\end{defn}

\begin{thm}
For a system of DDEs $\mathcal{A}=\{F_{\alpha}(x,n,[u])=0\}$, its linearized symmetry condition for an evolutionary generator $X=Q^{\alpha}(x,n,[u])\partial_{u^{\alpha}}$ is exactly  the Fr\'echet derivative of $F$ in the direction $Q$, namely
\begin{equation}
\bold{D}_F(Q)=0
\end{equation}
on all solutions.
\end{thm}
\begin{proof}
It immediately follows from the expansion below
\begin{equation}
\left(\bold{D}_F\right)_{\alpha\beta}(Q^{\beta})=\sum_{\beta,J_1,J_2}\frac{\partial F_{\alpha}}{\partial u_{J_1;J_2}^{\beta}}\left(D_{J_1}S_{J_2}Q^{\beta}\right)
=\bold{pr}X(F_{\alpha}).
\end{equation}
\end{proof}

\begin{defn}
A {\bf conservation law} for a system of DDEs $\mathcal{A}=\{F_{\alpha}(x,n,[u])=0\}$ with $x\in\mathbb{R}^{p_1}$ and $n\in\mathbb{Z}^{p_2}$ is a $(p_1;p_2)$-tuple $P(x,n,[u])=(P_1(x,n,[u]);P_2(x,n,[u]))$ subject to the following vanishment of {\bf divergence in the differential-difference sense}, that is
\begin{equation}
\operatorname{Div}P_1+\operatorname{Div}^{\vartriangle}P_2=0
\end{equation}
on all solutions of the system, where $\operatorname{Div}$ and $\operatorname{Div}^{\vartriangle}$ are the differential and difference divergence operators, respectively.
\end{defn}
Particularly in the $(1;p_2)$-dimensional case, namely $p_1=1$, a conservation law is
\begin{equation}\label{clsofdde}
D_tP_1+\operatorname{Div}^{\vartriangle}P_2=0,
\end{equation}
which vanishes on all solutions of a system of DDEs, where the function $P_1$ is often referred as a (conserved) density and the tuple $P_2$ is the (associated) flux. Under proper boundary conditions, the summation of the second term of \eqref{clsofdde} vanishes, leading to a conserved quantity given by the summation of $P_1$.  Computational examples for conservation laws of this type can be found in, for instance \cite{Goktas1997,Hickman2003}.

For DDEs $\mathcal{A}=\{F_{\alpha}(x,n,[u])=0\}$ of the bi-Kovalevskaya form, a conservation law can be written as
\begin{equation}
\operatorname{Div}\widehat{P}_1+\operatorname{Div}^{\vartriangle}\widehat{P}_2=\sum_{\alpha,J_1,J_2}K_{J_1,J_2}^{\alpha}(D_{J_1}S_{J_2}F_{\alpha})
\end{equation}
for some functions $K^{\alpha}_{J_1,J_2}(x,n,[u])$. This can be integrated and summed by parts to achieve 
\begin{equation}
\operatorname{Div}P_1+\operatorname{Div}^{\vartriangle}P_2=Q^{\alpha}(x,n,[u])F_{\alpha},
\end{equation}
which is an equivalent {\bf conservation law of the characteristic form}. The tuple of functions $Q(x,n,[u])$ is called its {\bf characteristic}. For interested readers, various methods for the computation of conservation laws for DDEs can be found in \cite{Goktas1997,Hickman2003,Zhang2003,Hydon2005}, for instance.


Now we are ready to compute the {\bf adjoint of the differential-difference Fr\'echet derivative} (\ref{ddfd})
via both integration by parts and summation by parts, which is
\begin{equation}\label{adfddd}
\left(\bold{D}_A\right)^{\ast}_{\alpha\beta}=\sum_{J_1,J_2}\left((-D)_{J_1}S_{-J_2}\right)\cdot\frac{\partial A^{\beta}}{\partial u^{\alpha}_{J_1;J_2}}.
\end{equation}

Consider a {\bf differential-difference Lagrangian} $L(x,n,[u])$. The underlying {\bf differential-difference Euler-Lagrange equations} are 
\begin{equation}
\bold{E}_{\alpha}(L(x,n,[u]))=0,
\end{equation}
where the {\bf differential-difference Euler operator} $\bold{E}$ is defined by
\begin{equation}
\bold{E}_{\alpha}:=\sum_{J_1,J_2}(-D)_{J_1}S_{-J_2}\frac{\partial }{\partial u^{\alpha}_{J_1;J_2}}.
\end{equation}

\begin{thm}\label{nulll}
A Lagrangian $L(x,n,[u])$ is a null Lagrangian such that $\bold{E}_{\alpha}(L)\equiv 0$ if and only if $L$ is a total differential-difference divergence 
\begin{equation}
L(x,n,[u])=\operatorname{Div}P_1+\operatorname{Div}^{\vartriangle}P_2
\end{equation}
for some $(p_1;p_2)$-tuple $(P_1(x,n,[u]);P_2(x,n,[u]))$.
\end{thm}
\begin{proof}
To prove $\bold{E}_{\alpha}(\operatorname{Div}P_1+\operatorname{Div}^{\vartriangle}P_2)\equiv 0$, we only need to recall that for fixed index $J_1$
\begin{equation}
\sum_{J_2}S_{-J_2}\frac{\partial }{\partial u^{\alpha}_{J_1;J_2}}(\operatorname{Div}^{\vartriangle}P_2)\equiv 0,
\end{equation}
and for fixed $J_2$
\begin{equation}
\sum_{J_1}(-D)_{J_1}\frac{\partial }{\partial u^{\alpha}_{J_1;J_2}}(\operatorname{Div}P_1)\equiv 0.
\end{equation}

The converse can be similarly proved following \cite{Olver1993}. Suppose $L(x,n,[u])$ is a null Lagrangian and consider the derivative
\begin{equation}
\frac{\operatorname{d}}{\operatorname{d}\!\varepsilon}L(x,n,[\varepsilon u])=\sum_{\alpha,J_1,J_2}u_{J_1;J_2}^{\alpha}\frac{\partial }{\partial u^{\alpha}_{J_1;J_2}}L(x,n,[\varepsilon u]).
\end{equation}
Each term can be re-arranged using integration by parts and summation by parts, and the above equality becomes
\begin{equation}\label{nlend}
\begin{aligned}
\frac{\operatorname{d}}{\operatorname{d}\!\varepsilon}L(x,n,[\varepsilon u])&=\sum_{\alpha}u^{\alpha}\sum_{J_1,J_2}(-D)_{J_1}S_{-J_2}\frac{\partial }{\partial u^{\alpha}_{J_1;J_2}}L(x,n,[\varepsilon u])+\operatorname{Div}R_1+\operatorname{Div}^{\vartriangle}R_2\\
&=u^{\alpha}\bold{E}_{\alpha}(L)(x,n,[\varepsilon u])+\operatorname{Div}R_1+\operatorname{Div}^{\vartriangle}R_2
\end{aligned}
\end{equation}
for some $(p_1;p_2)$-tuple $(R_1(\varepsilon;x,n,[u]);R_2(\varepsilon;x,n,[u]))$. Since $L$ is a null Lagrangian and therefore $\bold{E}(L)\equiv 0$. Thus, we can integrate (\ref{nlend}) with respect to $\varepsilon$ from $0$ to $1$, 
\begin{equation}
L(x,n,[u])-L(x,n,[0])=\operatorname{Div}\widehat{R}_1+\operatorname{Div}^{\vartriangle}\widehat{R}_2,
\end{equation}
where ($k=1,2$) 
\begin{equation}
\widehat{R}_k(x,n,[u])=\int_0^1 R_k(\varepsilon;x,n,[u])\operatorname{d}\!\varepsilon.
\end{equation}
Viewing $n$ as a parameter, there always exist $p_1$ numbers of functions $B(x,n)$ such that 
\begin{equation}
\operatorname{Div}B(x,n)=L(x,n,[0]).
\end{equation}
Thus a null Lagrangian $L(x,n,[u])$ is in the total differential-difference divergence form
\begin{equation}
L(x,n,[u])=\operatorname{Div}P_1+\operatorname{Div}^{\vartriangle}P_2,
\end{equation}
where 
\begin{equation}
P_1(x,n,[u])=\widehat{R}_1+B(x,n),\quad P_2(x,n,[u])=\widehat{R}_2.
\end{equation}
\end{proof}

\begin{thm}\label{leib}
For two $r$-tuples $A(x,n,[u])$ and $B(x,n,[u])$, we have that 
\begin{equation}
\bold{E}(A\cdot B)=\bold{D}_A^{\ast}(B)+\bold{D}_B^{\ast}(A).
\end{equation}
\end{thm}
\begin{proof}
This is the differential-difference correspondence to the results in Remark \ref{ccld} (the differential version) and Lemma \ref{lempeng} (the difference version). It can be immediately proved using the adjoint operator of the Fr\'echet derivative (\ref{adfddd}) and the Leibniz rule
\begin{equation}
\bold{E}_{\alpha}(A\cdot B)=\sum_{\beta,J_1,J_2}(-D)_{J_1}S_{-J_2}\left(\frac{\partial A^{\beta}}{\partial u^{\alpha}_{J_1;J_2}}B^{\beta}+\frac{\partial B^{\beta}}{\partial u^{\alpha}_{J_1;J_2}}A^{\beta}\right).
\end{equation}
\end{proof}

\begin{cor}
Let $\mathcal{A}=\{F_{\alpha}(x,n,[u])=0\}$ be a system of DDEs. Let $(P_1;P_2)$ be a conservation law of the characteristic form:
\begin{equation}
\operatorname{Div}P_1+\operatorname{Div}^{\vartriangle}P_2=Q^{\alpha}F_{\alpha}.
\end{equation}
Then the characteristic $Q$ satisfies  
\begin{equation}
\bold{D}_F^{\ast}(Q)=0
\end{equation}
on all solutions of the system $\mathcal{A}$.
\end{cor}
\begin{proof}
From Theorem \ref{nulll}, because $Q\cdot F$ is in the total divergence form, we have $\bold{E}(Q\cdot F)\equiv 0$. Then Theorem \ref{leib} gives
\begin{equation}
0=\bold{D}_Q^{\ast}(F)+\bold{D}_F^{\ast}(Q).
\end{equation}
Further, $\bold{D}_Q^{\ast}(F)$ always vanishes on all solutions of the system $\mathcal{A}$ as it reads
\begin{equation}
\bold{D}_Q^{\ast}(F)=\sum_{\beta,J_1,J_2}(-D)_{J_1}S_{-J_2}\left(\frac{\partial Q^{\beta}}{\partial u_{J_1;J_2}}F_{\beta}\right).
\end{equation}
This completes the proof. 
\end{proof}

\begin{defn}[Variational symmetry criterion]
A regular vector field 
\begin{equation}
X=\xi^{i}(x)\partial_{x^i}+\phi^{\alpha}(x,n,[u])\partial_{u^{\alpha}}
\end{equation}
generates a group of {\bf (divergence) variational symmetries} for a Lagrangian $L(x,n,[u])$ if there exists a $(p_1;p_2)$-tuple $(P_1(x,n,[u]);P_2(x,n,[u]))$ subject to 
\begin{equation}\label{dvsdde}
\bold{pr}X(L)+L(D_i\xi^i)=\operatorname{Div}P_1+\operatorname{Div}^{\vartriangle}P_2.
\end{equation}
\end{defn}
Taking the prolongation formula (\ref{proppro}) into consideration, equation (\ref{dvsdde}) is equivalent to the existence of a $(p_1;p_2)$-tuple $(P_1(x,n,[u]);P_2(x,n,[u]))$ such that 
\begin{equation}\label{evovs}
\bold{pr}X(L)=\operatorname{Div}P_1+\operatorname{Div}^{\vartriangle}P_2,
\end{equation}
where now $X=Q^{\alpha}(x,n,[u])\partial_{u^{\alpha}}$ is an evolutionary vector field. A practical method for calculating (divergence) variational symmetries is due to Theorem \ref{nulll}: Applying the differential-difference Euler operator to the equality \eqref{evovs}, we have
\begin{equation}\label{facteo}
\bold{E}\left(\bold{pr}X(L)\right)\equiv 0.
\end{equation}
To get the whole classification of (divergence) variational symmetries is often challenging; however, it can be simplified when particular ansatz is used, namely to search for symmetries of particular form.

%
%
%
%

\begin{thm}[Noether's theorem for DDEs]\label{NTDDE}
Consider a regular infinitesimal generator $X$ of a group of symmetries for a differential-difference variational problem with Lagrangian $L(x,n,[u])$:
\begin{equation}
X=\xi^i(x)\partial_{x^i}+\phi^{\alpha}(x,n,[u])\partial_{u^{\alpha}}.
\end{equation}
Its characteristic is 
\begin{equation}
Q^{\alpha}(x,n,[u])=\phi^{\alpha}-\xi^iu^{\alpha}_{\bold{1}_i;\bold{0}}.
\end{equation}
Then $Q$ is also a characteristic of a conservation law for the corresponding Euler-Lagrange equations. Namely there exists a $(p_1;p_2)$-tuple $(P_1(x,n,[u]);P_2(x,n,[u]))$ such that 
\begin{equation}
\operatorname{Div}P_1+\operatorname{Div}^{\vartriangle}P_2=Q^{\alpha}\bold{E}_{\alpha}(L).
\end{equation}
\end{thm}
\begin{proof}
Substituting the prolongation formula (\ref{proppro}) into the variational symmetry criterion (\ref{dvsdde}), we get that there exists some $(p_1;p_2)$-tuple $(\widehat{P}_1(x,n,[u]);\widehat{P}_2(x,n,[u]))$ such that 
\begin{equation}
\begin{aligned}
\operatorname{Div}\widehat{P}_1+\operatorname{Div}^{\vartriangle}\widehat{P}_2&=\bold{pr}X(L)+L(D_i\xi^i)\\
&=\sum_{\alpha,J_1,J_2}(D_{J_1}S_{J_2}Q^{\alpha})\frac{\partial L}{\partial u^{\alpha}_{J_1;J_2}}+\xi^i(D_iL)+L(D_i\xi^i)\\
&=\sum_{\alpha,J_1,J_2}(D_{J_1}S_{J_2}Q^{\alpha})\frac{\partial L}{\partial u^{\alpha}_{J_1;J_2}}+\operatorname{Div}(L\xi).
\end{aligned}
\end{equation}
The first term on the right hand side of the equality can be integrated and summed by parts:
\begin{equation}
\begin{aligned}
\sum_{\alpha,J_1,J_2}(D_{J_1}S_{J_2}Q^{\alpha})\frac{\partial L}{\partial u^{\alpha}_{J_1;J_2}}&=\sum_{\alpha,J_1,J_2}(S_{J_1}Q^{\alpha})(-D)_{J_1}\left(\frac{\partial L}{\partial u^{\alpha}_{J_1;J_2}}\right)+\operatorname{Div}R_1\\
&=\sum_{\alpha}Q^{\alpha}\sum_{J_1,J_2}S_{-J_2}(-D)_{J_1}\left(\frac{\partial L}{\partial u^{\alpha}_{J_1;J_2}}\right)+\operatorname{Div}R_1+\operatorname{Div}^{\vartriangle}R_2\\
&=Q^{\alpha}\bold{E}_{\alpha}(L)+\operatorname{Div}R_1+\operatorname{Div}^{\vartriangle}R_2.
\end{aligned}
\end{equation}
Therefore, we obtain that
\begin{equation}
\operatorname{Div}P_1+\operatorname{Div}^{\vartriangle}P_2=Q^{\alpha}\bold{E}_{\alpha}(L),
\end{equation}
where 
\begin{equation}
P_1=\widehat{P}_1-R_1-L\xi,\quad P_2=\widehat{P}_2-R_2.
\end{equation}
\end{proof}

In the following examples, unless otherwise specified, we let $t$ be the continuous independent variable and $n$ be the discrete independent variable. Both of them are assumed to be one-dimensional. 

\begin{exm}
Consider the following Lagrangian 
\begin{equation}
L=-\frac{(u')^2}{2}+\frac{au^2}{2}+\frac{b+cn}{2}\left(u_1-u\right)^2,
\end{equation}
where $a,b,c$ are positive constants. Its Euler-Lagrange equation is
\begin{equation}
u''+au-(b+cn)(u_1-u)+(b+c(n-1))(u-u_{-1})=0,
\end{equation}
which describes a small vibration of a compound pendulum consisting of a light string and a large mass at the end (e.g. \cite{Bateman1943}). One can immediately check that $\cos (\sqrt{a}t)$, $\sin(\sqrt{a}t)$ and $u'$ are characteristics of divergence variational symmetries, for example, by checking $\bold{E}(\bold{pr}X(L))=0$  (i.e. equality \eqref{facteo}), where $X=Q(t,n,[u])\partial_u$ and $Q$ are the characteristics. As a consequence of Theorem \ref{NTDDE}, they contribute to three characteristics of conservation laws for the Euler-Lagrange equation. They are
\begin{equation}
\begin{aligned}
D_t\{\cos(\sqrt{a}t)u'+\sqrt{a}\sin(\sqrt{a}t)u\}+(S-\operatorname{id})\{-\cos(\sqrt{a}t)(b+c(n-1))(u-u_{-1})\}&=\cos(\sqrt{a}t)\bold{E}(L),\\
D_t\{\sin(\sqrt{a}t)u'-\sqrt{a}\cos(\sqrt{a}t)u\}+(S-\operatorname{id})\{-\sin(\sqrt{a}t)(b+c(n-1))(u-u_{-1})\}&=\sin(\sqrt{a}t)\bold{E}(L),\\
D_t\left\{\frac{(u')^2}{2}+\frac{au^2}{2}+\frac{b+cn}{2}(u_1-u)^2\right\}+(S-\operatorname{id})\{-u'(b+c(n-1))(u-u_{-1})\}&=u'\bold{E}(L).
\end{aligned}
\end{equation}
\end{exm}

\begin{exm}
Consider the Toda lattice (e.g. \cite{Flaschka1974a,Flaschka1974b,Toda1970})
\begin{equation}
u''+\exp(u-u_1)-\exp(u_{-1}-u)=0,
\end{equation}
which is one of the well-known discretisations for the KdV equation. It is integrable and admits infinitely many symmetries \cite{Yamilov2006}. Here we consider its conservation laws using the governing Lagrangian 
\begin{equation}
L=-\frac{(u')^2}{2}+\exp(u-u_1).
\end{equation}
It admits variational symmetries with the following characteristics 
\begin{equation}
Q_1=1,~Q_2=t,~Q_3=u',
\end{equation}
leading to conservation laws
\begin{equation}
\begin{aligned}
D_t(u')+(S-\operatorname{id})\exp(u_{-1}-u)&=\bold{E}(L),\\
D_t(tu'-u)+(S-\operatorname{id})(t\exp(u_{-1}-u))&=t\bold{E}(L),\\
D_t\left(\frac{(u')^2}{2}+\exp(u-u_1)\right)+(S-\operatorname{id})(u'\exp(u_{-1}-u))&=u'\bold{E}(L).
\end{aligned}
\end{equation}
\end{exm}

\begin{exm}
The KdV equation 
\begin{equation}
u_t+uu_x+u_{xxx}=0
\end{equation} 
can be rewritten as
\begin{equation}
v_{tx}+v_xv_{xx}+v_{xxxx}=0,
\end{equation} 
introducing $v_x=u$. The latter is governed by a Lagrangian 
\begin{equation}\label{lofkdv}
L=-\frac{v_tv_x}{2}-\frac{v_x^3}{6}+\frac{v_{xx}^2}{2},
\end{equation}
which admits the following symmetries 
\begin{equation}
Q_1=1,~Q_2=v_x,~Q_3=v_x^2+2v_{xxx},~Q_4=t.
\end{equation}
Hence they contribute to four distinct conservation laws. The first three can be changed back to conservation laws of the original equation using the same transformation $v_x=u$ and they become
\begin{equation}\label{clkdv}
\begin{aligned}
D_tu+D_x\left(\frac{1}{2}u^2+u_{xx}\right)&=F,\\
D_t\left(\frac{1}{2}u^2\right)+D_x\left(\frac{1}{3}u^3+uu_{xx}-\frac{1}{2}u_x^2\right)&=uF,\\
D_t\left(\frac{1}{3}u^3-u_x^2\right)+D_x\left(\frac{1}{4}u^4+u^2u_{xx}+2u_xu_t+u_{xx}^2\right)&=\left(u^2+2u_{xx}\right)F,
\end{aligned}
\end{equation}
where $F=u_t+uu_x+u_{xxx}$. However, the last one with characteristic $Q_4=t$ can not be transformed back because its flux depends on $v$. These conservation laws \eqref{clkdv} are respectively (equivalent to) the conservation of mass, the conservation of momentum and the conservation of energy (e.g. \cite{DJ1989}).  

Next we consider semi-discretisations of the KdV equation. As we will see, by properly choosing semi-discretisations, we may preserve multiple symmetries and/or multiple conservation laws simultaneously.

We start with semi-discretisations of the Lagrangian \eqref{lofkdv}, for instance
\begin{equation}
L_1=-\frac{v'}{2}(v_1-v)-\frac{(v_1-v)^3}{6}+\frac{(v_1-2v+v_{-1})^2}{2}.
\end{equation}
Now $v'=v_t$. The underlying DDE (i.e. the Euler-Lagrange equation $\bold{E}(L_1)=0$) is
\begin{equation}
\frac{v_1'-v_{-1}'}{2}+\frac{(v_1-v)^2-(v-v_{-1})^2}{2}+v_2-4v_1+6v-4v_{-1}+v_{-2}=0.
\end{equation}
It becomes a semi-discretisation of the original KdV equation, introducing $v-v_{-1}=u$, and it reads
\begin{equation}\label{ddekdv1}
\frac{u'_1+u'}{2}+\frac{u_1^2-u^2}{2}+u_2-3u_1+3u-u_{-1}=0.
\end{equation}
In this case, symmetries with characteristics $Q_1=1$ and $Q_4=t$ are preserved, namely they are still variational symmetries of $L_1$ and hence contributes to conservation laws of the Euler-Lagrange equation. The first one becomes a conservation law of the semi-discretised equation \eqref{ddekdv1}:
\begin{equation}
D_t\left(\frac{u_1+u}{2}\right)+(S-\operatorname{id})\left(\frac{1}{2}u^2+u_1-2u+u_{-1}\right)=F_1,
\end{equation}
where $F_1$ is the left hand side of \eqref{ddekdv1}.

On the other side, let us consider semi-discretisations by discretising time $t$. For instance, consider the following differential-difference Lagrangian
\begin{equation}
L_2=-\frac{v_1-v}{2}\frac{v_1'+v'}{2}-\frac{(v')^3}{6}+\frac{(v'')^2}{2}.
\end{equation}
Now `dash' denotes derivatives with respect to $x$, for example $v'=v_x$ and so forth, while $n$ is the discretised time. Its Euler-Lagrange equation is 
\begin{equation}
\frac{v_1'-v_{-1}'}{2}+v'v''+v''''=0,
\end{equation}
which becomes a semi-discretisation of the original KdV equation using $v'=u$, namely 
\begin{equation}\label{kdvdde2}
\frac{u_1-u_{-1}}{2}+uu'+u'''=0.
\end{equation}
Now symmetries with characteristics $Q_1,Q_2,Q_4$ are preserved and they become 
\begin{equation}
Q_1=1,~~Q_2=v',~~Q_4=n.
\end{equation} 
They yield three conservation laws of the Euler-Lagrange equation; the first two become conservation laws of the DDE \eqref{kdvdde2}:
\begin{equation}
\begin{aligned}
(S-\operatorname{id})\left(\frac{u_1+u}{2}\right)+D_x\left(\frac{1}{2}u^2+u''\right)&=F_2,\\
(S-\operatorname{id})\left(\frac{uu_{-1}}{2}\right)+D_x\left(\frac{1}{3}u^3+uu''-\frac{1}{2}(u')^2\right)&=uF_2,\\
\end{aligned}
\end{equation}
Here $F_2$ is the left hand side of \eqref{kdvdde2}.

\end{exm}

\begin{exm}\label{vsve}
Consider the Volterra equation (Example \ref{volterra}) again. Introduce a new variable via
\begin{equation}
u=\exp(v_1-v_{-1}),
\end{equation}
and we have a new differential-difference equation
\begin{equation}
v_1'-v_{-1}'=\exp(v_2-v)-\exp(v-v_{-2}),
\end{equation}
which admits a differential-difference Lagrangian 
\begin{equation}
L=v(v_1'-v')+\exp(v_2-v).
\end{equation}
The following Table \ref{clv} includes several regular (evolutionary) variational symmetries $X=Q\partial_v$ and corresponding conservation laws
\begin{equation}
D_tP_1+(S-\operatorname{id})P_2=Q\bold{E}(L).
\end{equation}

\begin{table}[h]
\caption{Some conservation laws of the Volterra equation}\label{clv} 
\centering 
\begin{tabular}{l l l l} 
\vspace{0.05cm}
 \textbf{Characteristics} & \textbf{Conservation laws} \\ [0.5ex] 
\hline 
$Q=1$ & $P_1=v_1-v_{-1}$\vspace{0.08cm}\\
& $P_2=-\exp(v_1-v_{-1})-\exp(v-v_{-2})$\\[1ex]
\hline
$Q=(-1)^n$ & $P_1=(-1)^n(v_1-v_{-1})$\vspace{0.08cm}\\
 &  $P_2=(-1)^n\exp(v_2-v)-(-1)^n\exp(v-v_{-2})$\\
 [1ex]
 \hline
 $Q=f(t)$ &$P_1=0$\vspace{0.08cm}\\
 & $P_2=f(t)\left(v'+v_{-1}'-\exp(v_1-v_{-1})-\exp(v-v_{-2})\right)$\\
[1ex] 
\hline 
\end{tabular}
\label{table:nonlin} 
\end{table}
\end{exm}

There exist many other integrable DDEs, which are variational either as what they stand or after introducing a change of variables. For instance,
\begin{itemize}
\item The modified Volterra equation (e.g. \cite{Yamilov2006,Khanizadeh2013})
\begin{equation}
u'=u^2(u_1-u_{-1}).
\end{equation}
A change of variables is
\begin{equation}
u=\frac{1}{v_1-v_{-1}};
\end{equation}
 the Lagrangian is
\begin{equation}
L=v(v'_1-v')-\ln (v_2-v).
\end{equation}

\item The Ablowitz-Ramani-Segur (Gerdjikov-Ivanov) lattice equation (e.g. \cite{Khanizadeh2013,Tsu2002})
\begin{equation}
\left\{
\begin{array}{l}
u'=(au_1-bu_{-1})(1+uv)(1-uv_1)\vspace{0.25cm}\\
v'=(bv_1-av_{-1})(1+uv)(1-u_{-1}v).
\end{array}
\right.
\end{equation}
Itself is variational governed by the Lagrangian 
\begin{equation}
 L = \ln((1+ u v) (1- u v_1)) \frac{u'}{u}  - a ( u v_{-1} - u v - u u_1 v v_1 )- b ( u_{-1} v_1 - u v_1 + u_{-1} u v v_1 ).
\end{equation}
\item The Kaup-Newell lattice equation (e.g. \cite{Khanizadeh2013,Tsu2002})
\begin{equation}
\left\{
\begin{array}{l}
u'=a\left(\frac{u_1}{1-u_1v_1}-\frac{u}{1-uv}\right)+b\left(\frac{u}{1+uv_1}-\frac{u_{-1}}{1+u_{-1}v}\right)\vspace{0.25cm}\\
v'=a\left(\frac{v}{1-uv}-\frac{v_{-1}}{1-u_{-1}v_{-1}}\right)+b\left(\frac{v_1}{1+uv_1}-\frac{v}{1+u_{-1}v}\right).
\end{array}
\right.
\end{equation}
A change of variables is
\begin{equation}
u=\mu_1-\mu,\quad v=\nu-\nu_{-1};
\end{equation}
 the Lagrangian is
\begin{equation}
L=\nu(\mu_1'-\mu')-a\ln\left\{1-(\mu_2-\mu_1)(\nu_1-\nu)\right\}+b\ln\left\{1+(\mu_1-\mu)(\nu_1-\nu)\right\}.
\end{equation}
\end{itemize}



\section{Self-adjointness and conservation laws of non-variational differential-difference equations}

\label{section4}

In many situations, people are more interested at symmetries and conservation laws for  integrable DDEs themselves rather than the governing variational principles in terms of other variables. The applicability of Noether's theorem relies on variational symmetries rather than symmetries of the original integrable DDEs, while the latter, however, are often better known. Is it possible to construct conservation laws directly from symmetries of the original DDEs without introducing new variables? The self-adjointness method answers this question and we construct the differential-difference version in this section. The reader should refer to \cite{Ibragimov2007} for the differential version and \cite{Peng2015} for the difference version.

Consider a system of DDEs
\begin{equation}
\mathcal{A}=\{F_{\alpha}(x,n,[u])=0\}
\end{equation}
and its regular symmetry generators 
\begin{equation}
X=\xi^i(x)\partial_{x^i}+\phi^{\alpha}(x,n,[u])\partial_{u^{\alpha}}.
\end{equation}

\begin{defn}
The {\bf adjoint system} to the differential-difference system $\mathcal{A}$ is defined as
\begin{equation}
0=F_{\alpha}^{\ast}(x,n,[u],[v]):=\bold{E}_{u^{\alpha}}(v_n^{\beta}F_{\beta}).
\end{equation}
Here we introduce a new variable $v_n$, which is to be determined.  The function 
\begin{equation}
L=v_n^{\beta}F_{\beta}
\end{equation}
 is sometimes called a {\bf formal Lagrangian}.
\end{defn}

\begin{defn}
\label{sadde}
A system of DDEs $\mathcal{A}$ is said to be {\bf self-adjoint} if its adjoint system  
\begin{equation}
F_{\alpha}^{\ast}(x,n,[u],[v])=0
\end{equation}
holds on all solutions of the system $\mathcal{A}$ by some substitution $v^{\alpha}=f^{\alpha}(x,n,[u])$. 
\end{defn}
It is important to distinguish self-adjointness of a system of DDEs in Definition \ref{sadde} from self-adjointness of the associated Fr\'echet derivative; the former is weaker.

To be consistent with the differential and difference versions, a system of DDEs is, respectively, called strict, quasi and weak self-adjoint corresponding to the substitutions $v^{\alpha}=u^{\alpha}$, $v^{\alpha}=f^{\alpha}([u])$ and $v^{\alpha}=f^{\alpha}(x,n,[u])$. 

\begin{exm}\label{savol}
Again consider the running example, the Volterra equation in Example \ref{volterra}.
The formal Lagrangian is given by
\begin{equation}
L=v\left(\frac{u'}{u}-u_{1}+u_{-1}\right),
\end{equation}
the Euler-Lagrange equations of which provide the Volterra equation itself and its adjoint equation
\begin{equation}
-\frac{v'}{u}+v_1-v_{-1}=0.
\end{equation}
The adjoint equation becomes the Volterra equation itself by the substitution $v=-u$. Hence the Volterra equation is self-adjoint.
\end{exm}

\begin{exm}
Consider the following integrable DDE classified by Yamilov \cite{Yamilov2006}:
\begin{equation}
u'=\frac{1}{u_1-u}+\frac{1}{u-u_{-1}}.
\end{equation}
Its adjoint equation is
\begin{equation}
-v'-\frac{v_1+v}{(u_1-u)^2}+\frac{v+v_{-1}}{(u-u_{-1})^2}=0.
\end{equation}
Substituting $v=(-1)^nu$ inside, the adjoint equation becomes
\begin{equation}
(-1)^{n+1}\left(u'-\frac{1}{u_1-u}-\frac{1}{u-u_{-1}}\right)=0,
\end{equation}
which is equivalent to the original equation. 
\end{exm}

The next theorem is essential in constructing conservation laws for non-variational but self-adjoint DDEs  from their symmetries. The differential and difference counterparts can be respectively found in \cite{Ibragimov2007} and \cite{Peng2015}.

\begin{thm}
 Any regular symmetries of a system of DDEs of the bi-Kovalevskaya form can be extended to variational symmetries of the formal Lagrangian. 
\end{thm} 

\begin{proof}
Let $\mathcal{A}=\{F_{\alpha}(x,n,[u])=0\}$ be a system of DDEs of the bi-Kovalevskaya form. To prove the theorem, it is enough to consider a regular symmetry generator in the evolutionary form (see Remark \ref{evvf})
\begin{equation}
X=Q^{\alpha}(x,n,[u])\partial_{u^{\alpha}}.
\end{equation}
Consider an extended evolutionary vector field
\begin{equation}\label{xtoy}
Y=X+Q_{\ast}^{\alpha}(x,n,[u],[v])\partial_{v^{\alpha}}
\end{equation}
for some functions $Q^{\alpha}_*$ to be determined. Next we consider the infinitesimal invariance principle (\ref{dvsdde}) for the formal Lagrangian $L=v^{\alpha}F_{\alpha}$. Since the system is in the bi-Kovalevskaya form, there exist functions $K^{\beta}_{\alpha;J_1,J_2}(x,n,[u])$ such that 
\begin{equation}\label{eqfl13}
\begin{aligned}
\bold{pr} Y(L)&=Q^{\alpha}_{\ast}F_{\alpha}+v^{\alpha}\bold{pr}X(F_{\alpha})\\
&=Q^{\alpha}_{\ast}F_{\alpha}+\sum_{\alpha,\beta,J_1,J_2}v^{\alpha}K^{\beta}_{\alpha;J_1,J_2}(D_{J_1}S_{J_2}F_{\beta})\\
&=Q_{\ast}^{\beta}F_{\beta}+\sum_{\alpha,\beta,J_1,J_2}(-D)_{J_1}S_{-J_2}\left(v^{\alpha}K^{\beta}_{\alpha;J_1,J_2}\right)F_{\beta}+\operatorname{Div}P_1+\operatorname{Div}^{\vartriangle}P_2\\
&=\left(Q_{\ast}^{\beta}+\sum_{\alpha,J_1,J_2}(-D)_{J_1}S_{-J_2}\left(v^{\alpha}K^{\beta}_{\alpha;J_1,J_2}\right)\right)F_{\beta}+\operatorname{Div}P_1+\operatorname{Div}^{\vartriangle}P_2
\end{aligned}
\end{equation}
for some $(p_1;p_2)$-tuple $(P_1(x,n,[u],[v]);P_2(x,n,[u],[v]))$. Hence the extended evolutionary vector field $Y$ is a variational symmetry for the formal Lagrangian if we define
\begin{equation}\label{qstar}
Q_{\ast}^{\beta}(x,n,[u],[v])=-\sum_{\alpha,J_1,J_2}(-D)_{J_1}S_{-J_2}\left(v^{\alpha}K^{\beta}_{\alpha;J_1,J_2}\right).
\end{equation}
This completes the proof.
\end{proof}

The above theorem allows us to apply Noether's theorem for DDEs to compute conservation laws for the Euler-Lagrange equations governed by a formal Lagrangian. In particular, the resulting conservation laws with arguments $(x,n,[u],[v])$  become conservation laws of the original system via a certain substitution $v^{\alpha}=f^{\alpha}(x,n,[u])$, for which the original system is self-adjoint.

\begin{exm}\label{exmsadde}
Consider the Volterra equation, which is self-adjoint as shown in Example \ref{savol}.  It is  well known that it admits infinitely many  symmetries. In particular, its Lie point symmetries have been calculated in Example \ref{volterra}. Take the regular infinitesimal generator $X_1=-t\partial_t+u\partial_u$ for example. Now the characteristic is $Q=u+tu'$. From (\ref{qstar}), we obtain $Q_{\ast}=tv'$. A conservation law $(\widehat{P}_1;\widehat{P}_2)$ can then be obtained  in the characteristic form:
\begin{equation}
D_t\widehat{P}_1+(S-\operatorname{id})\widehat{P}_2=Q\left(-\frac{v'}{u}+v_1-v_{-1}\right)+Q_{\ast}\left(\frac{u'}{u}-u_1+u_{-1}\right).
\end{equation}
Substituting $v=-u$ inside, this becomes a conservation law of the Volterra equation in the characteristic form:
\begin{equation}
D_tP_1+(S-\operatorname{id})P_2=u\left(\frac{u'}{u}-u_1+u_{-1}\right).
\end{equation}
One choice of the components is
\begin{equation}
P_1=u,\quad P_2=-uu_{-1}.
\end{equation}
\end{exm}

\begin{rem}\label{withouty}
In Example \ref{exmsadde} above, we first extended symmetries of a system DDEs to variational symmetries of the formal Lagrangian and then apply Noether's theorem. However, it is not necessary to calculate the extended variational symmetries explicitly. 

Again consider a system of DDEs $\mathcal{A}=\{F_{\alpha}(x,n,[u])=0\}$ and its infinitesimal generator $X$, which is assumed in the evolutionary form $X=Q^{\alpha}(x,n,[u])\partial_{u^{\alpha}}$. Let $Y=X+Q_{\ast}^{\alpha}(x,n,[u],[v])\partial_{v^{\alpha}}$ be the extended symmetry generator of the formal Lagrangian $L=v^{\alpha}F_{\alpha}$. Recall that the following equality holds (equation \eqref{eqfl13})
\begin{equation}
Q^{\alpha}_{\ast}F_{\alpha}=-v^{\alpha}\bold{pr}X(F_{\alpha})+\operatorname{Div}P_1+\operatorname{Div}^{\vartriangle}P_2
\end{equation}
for some $(p_1;p_2)$-tuple $(P_1(x,n,[u],[v]);P_2(x,n,[u],[v]))$, which is clearly a trivial conservation law of the original system of DDEs for arbitrary functions $v^{\alpha}(x,n,[u])$. The conservation law obtained from Noether's theorem becomes
\begin{equation}
\begin{aligned}
\operatorname{Div}R_1+\operatorname{Div}^{\vartriangle}R_2&=Q^{\alpha}F_{\alpha}^{\ast}+Q^{\alpha}_{\ast}F_{\alpha}\\
&=Q^{\alpha}F_{\alpha}^{\ast}-v^{\alpha}\bold{pr}X(F_{\alpha})+\operatorname{Div}P_1+\operatorname{Div}^{\vartriangle}P_2.
\end{aligned}
\end{equation}
This can be further written as an equivalent conservation law as
\begin{equation}\label{eqclsf}
\operatorname{Div}\widehat{P}_1+\operatorname{Div}^{\vartriangle}\widehat{P}_2=Q^{\alpha}F_{\alpha}^{\ast}-v^{\alpha}\bold{pr}X(F_{\alpha}),
\end{equation}
where $\widehat{P}_i(x,n,[u],[v])=R_i(x,n,[u],[v])-P_i(x,n,[u],[v])$ ($i=1,2$). With the proper substitution $v^{\alpha}=f^{\alpha}(x,n,[u])$ such that the system is self-adjoint, this becomes a conservation law of the original system. This approach applies to differential equations and difference equations in the same manner.
\end{rem}

For instance, we can consider Example \ref{exmsadde} in the following way without calculating $Q_{\ast}$; now $X_1=-t\partial_t+u\partial_u$, which is equivalent to one generator in the evolutionary form with characteristic $Q=u+tu'$. The conservation law is (according to \eqref{eqclsf}) 
\begin{equation}
\begin{aligned}
D_tP_1+(S-\operatorname{id})P_2&=(u+tu')\left(\frac{u'}{u}-u_1+u_{-1}\right)+u \cdot \bold{pr}\left(Q\partial_u\right)\left(\frac{u'}{u}-u_1+u_{-1}\right)\\
&=(u+tu')\left(\frac{u'}{u}-u_1+u_{-1}\right)+(u+tuD_t)\left(\frac{u'}{u}-u_1+u_{-1}\right)\\
&=u\left(\frac{u'}{u}-u_1+u_{-1}\right)+D_t\left(tu\left(\frac{u'}{u}-u_1+u_{-1}\right)\right).
\end{aligned}
\end{equation}
Therefore, the same result is obtained as in Example \ref{exmsadde} that $u$ is a characteristic of conservation law for the Volterra equation. Here we used the substitution $v=-u$.

\section{Conclusion}
Noether's theorem is a celebrated result which establishes a relation between variational symmetries and conservation laws of the underlying Euler-Lagrange equations; the differential and difference versions have already been well studied. As the first main result of this paper, we extended Noether's theorem to differential-difference equations (DDEs).  Prolongation formulae of continuous symmetries were first investigated and the equivalence of regular symmetry generators and evolutionary vector fields was understood. The latter allows us to connect variational symmetries and conservation laws through their characteristics. For non-variational DDEs, we adapted the self-adjointness method for calculating conservation laws, which had been extensively studied for differential equations during the last decade. Its extension to difference equations was made by the author in \cite{Peng2015}. Defining a formal Lagrangian, the self-adjointness method allows one to achieve infinitely many conservation laws directly from infinitely many (non-variational) symmetries. Though further work may be needed to examine if the resulting conservation laws are distinct. Illustrative examples were provided. 

%
%

\section*{ Acknowledgements} 
The author is indebted to Cheng Zhang and Da-jun Zhang for their hospitality during his visit at Shanghai University, when part of this work was done. The author is grateful to Peter Hydon, whose comments on \cite{Peng2015}  led to Theorem \ref{phydon}. The author would also like to thank Pavlos Xenitidis for insightful discussions on variational principle of DDEs. This work was partially supported by Grant-in-Aid for Scientific Research (16KT0024) and Waseda University Grants for Special Research Projects (2016B-119).


\end{document}